\PassOptionsToPackage{final}{graphicx}
\PassOptionsToPackage{colorlinks,linkcolor={blue},citecolor={blue},urlcolor={red},breaklinks=true,final}{hyperref}
\PassOptionsToPackage{nosumlimits,nonamelimits}{amsmath}

\documentclass[a4paper,UKenglish, autoref, thm-restate,authorcolumns,final]{llncs}

\usepackage{booktabs} 
\usepackage{tikz}
\usepackage{amsfonts}
\usepackage{amssymb}
\usepackage{stmaryrd}
\usepackage{adjustbox}
\usepackage{xspace}
\usepackage{wrapfig}
\usepackage[nice]{nicefrac}
\usepackage{mathtools}
\usepackage{yhmath}

\usepackage{bbold}
\usepackage{bm}

\usepackage{amssymb}
\usepackage{bbm}
\usepackage{microtype}
\usepackage{needspace}

\usepackage[font=small,skip=10pt]{caption}
\usetikzlibrary{shapes.geometric, arrows}
\usetikzlibrary{decorations.pathmorphing}
\usetikzlibrary{snakes}
\usetikzlibrary{patterns}
\usetikzlibrary{automata, positioning, arrows}
\usetikzlibrary{backgrounds}
\usetikzlibrary{arrows}
\usetikzlibrary{fit,calc}

\usepackage[capitalize,nameinlink]{cleveref}

\usepackage{savesym}

\savesymbol{degree}
\savesymbol{leftmoon}
\savesymbol{rightmoon}
\savesymbol{fullmoon}
\savesymbol{newmoon}
\savesymbol{diameter}
\savesymbol{emptyset}
\savesymbol{bigtimes}
\savesymbol{triangleright}

\savesymbol{langle}
\savesymbol{rangle}

\usepackage[matha]{mathabx}

\restoresymbol{other}{langle}
\restoresymbol{other}{rangle}
\restoresymbol{other}{triangleright}
\restoresymbol{other}{emptyset}

\usepackage{todos}

\usepackage{hypcap}
\usepackage{ifdraft}                            
\ifdraft
{

  \usepackage[layout=footnote,draft]{fixme}
    
  \usepackage[notcite,notref]{showkeys}
  
}
{
  \usepackage[layout=footnote,final]{fixme}
}

\FXRegisterAuthor{sg}{asg}{\color{green!50!black}SG} 
\FXRegisterAuthor{jp}{ajp}{\color{red!75!black}JP} 
\FXRegisterAuthor{rn}{arn}{\color{blue}RN} 

\newcommand{\eg}{e.g.}
\newcommand{\ie}{i.e.\xspace} 

\newcommand{\prog}[1]{\ensuremath{\tt #1}\xspace}
\newcommand{\until}{\textcolor{blue}{for}\xspace}

\newcommand{\uuntil}{\textcolor{blue}{until}}

\newcommand{\blue}[1]{\textcolor{blue}{#1}}
\newcommand{\progife}[3]{\prog{\blue{if} \> #1 \> \blue{then} \> #2 \> \blue{else} \> #3}}
\newcommand{\progwhile}[2]{\prog{\blue{while} \>  #1 \> \blue{do} \> \{ \> #2  \> \}}}

\newcommand{\tool}[1]{\textsc{#1}\xspace} 
\newcommand{\lrule}[3]{\textbf{#1}\quad\frac{#2}{#3}}

\newcommand{\infrule}[2]{\frac{#1}{#2}}

\newcommand{\tpl}[1]{\,#1\,} 

\newcommand{\inr}{\operatorname{\mathsf{inr}}}
\newcommand{\inl}{\operatorname{\mathsf{inl}}}
\newcommand{\inj}{\operatorname{\mathsf{inj}}}

\newcommand{\id}{\operatorname{\mathsf{id}}}

\newcommand{\snd}{\operatorname{\mathsf{snd}}}

\newcommand{\Lince}{\tool{Lince}}
\newcommand{\Reals}{\mathbb{R}}

\newcommand{\comp}{\cdot}
\newcommand{\scomp}{\, \blue{;} \,}

\providecommand{\comma}{,\operatorname{}\linebreak[1]}		
\newcommand{\ass}{\mathrel{\coloneqq}}
\newcommand{\sep}{\kern1pt\comma\kern-1pt}
\newcommand{\ssto}[1][]{~\to^{#1}~}
\newcommand{\bsto}{~\Downarrow~}
\newcommand{\stp}{\mathit{stop}}
\newcommand{\skp}{\mathit{skip}}
\newcommand{\Vars}{\mathcal{X}}
\newcommand{\brks}[1]{\langle #1\rangle}
\newcommand{\tconc}{\mathbin{\mbox{\,$\wideparen{}\,$}}}

\providecommand{\lsem}{\llbracket}
\providecommand{\rsem}{\rrbracket}
\providecommand{\sem}[1]{\lsem #1 \rsem}

\renewcommand{\l}{\lambda}

\let\cedilla\c
\renewcommand{\c}{\colon}

\newcommand{\om}{\mathbb{M}}
\newcommand{\mm}{\mathbb{E}}
\newcommand{\mmm}{\mathbb{F}}

\newcommand{\later}{\operatorname{\,\triangleright\,}}

\newcommand{\iso}{\mathbin{\cong}}
\newcommand{\klstar}{\star}  			
\newcommand{\istar}{\dagger}  			
\newcommand{\iistar}{\ddagger}  			

\newcommand{\by}[1]{\text{/\!\!/~#1}}			             

\makeatletter
\newcommand{\displaybump}{\hspace{3.5em}\hbox to -\@totalleftmargin{\hfil}}
\makeatother

\newcommand{\fot}{\oname{Trj}}

\newcommand{\bang}{\operatorname!}				             

\newcommand{\appr}{\sqsubseteq}
\newcommand{\lub}{\bigsqcup}
\providecommand{\bigor}{\bigvee}

\providecommand{\mto}{\mapsto}
\providecommand{\xto}[1]{\,\xrightarrow{#1}\,}

\providecommand{\mplus}{{\scriptscriptstyle\bf+}} 	       

\newcommand{\real}{\mathbb{R}}

\newcommand{\realp}{\real_{\mplus}}
\newcommand{\realpe}{\bar\real_{\mplus}}

\newcommand{\eps}{\operatorname\varepsilon}

\makeatletter
\def\mfix#1{\oname{#1}\@ifnextchar\bgroup\@mfix{}}	       
\def\@mfix#1{#1\@ifnextchar\bgroup\mfix{}}			           
\makeatother

\newcommand{\oname}[1]{\operatorname{\mathsf{#1}}}

\newcommand{\dr}{{\operatorname{\mathsf{d}}}} 
\newcommand{\ev}{{\operatorname{\mathsf{e}}}} 
\newcommand{\hc}{\textsc{Hyb\-Core}}

\newcommand{\ite}[3]{\mfix{\kern-2pt}{#1}{\,\raisebox{-0.2ex}{\scalebox{.6}[1.45]{$\lhd$}}}{\,#2\,}{\raisebox{-0.2ex}{\scalebox{.6}[1.45]{$\rhd$}}}{\,\mathbin{}#3}}

\def\defbbname#1{\expandafter\def\csname BB#1\endcsname{{\bm{\mathsf{#1}}}}}
\def\defbbnames#1{\ifx#1\defbbnames\else\defbbname#1\expandafter\defbbnames\fi}
\defbbnames ABCDEFGHIJKLMNOPQRSTUVWXYZ\defbbnames

\newcommand{\argument}{\operatorname{\mbox{$-\!-$}}}

\newcommand{\myparagraph}[1]{\medskip\noindent\textbf{#1.}~~}

\title{Implementing Hybrid Semantics: From Functional to Imperative} 
\author{Sergey Goncharov\inst{1}, Renato Neves\inst{2} and Jos\'{e} Proen\cedilla{c}a\inst{3}}
\institute{Dept.~of Comp.~Sci., FAU Erlangen-Nürnberg, Germany \and
University of Minho \& INESC-TEC, Portugal \and
CISTER/ISEP, Portugal}

\begin{document}

\maketitle
\allowdisplaybreaks

\begin{abstract}
  Hybrid programs combine digital control with differential equations,
  and naturally appear in a wide range of application domains, from
  biology and control theory to real-time software engineering. The
  entanglement of discrete and continuous behaviour inherent to such
  programs goes beyond the established computer science foundations,
  producing challenges related to \eg\ infinite iteration and
  combination of hybrid behaviour with other effects.
  A systematic treatment of \emph{hybridness} as a dedicated
  computational effect has emerged recently. In particular, a generic
  idealized functional language \hc{} with a sound and adequate
  operational semantics has been proposed. The latter semantics
  however did not provide hints to implementing \hc{} as a runnable
  language, suitable for hybrid system simulation (e.g.\ the semantics
  features rules with uncountably many premises).
  We introduce an imperative counterpart of \hc{}, whose semantics is
  simpler and runnable, and yet intimately related with the semantics
  of \hc{} at the level of \emph{hybrid monads}.  We then establish a
  corresponding soundness and adequacy theorem. To attest that the
  resulting semantics can serve as a firm basis for the implementation
  of typical tools of programming oriented to the hybrid domain, we
  present a web-based prototype implementation to evaluate and inspect
  hybrid programs, in the spirit of \tool{GHCi} for \tool{Haskell} and
  \tool{UTop} for \tool{OCaml}. The major asset of our implementation
  is that it formally follows the operational semantic rules.
\end{abstract}

\section{Introduction}\label{sec:intro}

\myparagraph{The core idea of hybrid programming} Hybrid programming
is a rapidly emerging computational
paradigm~\cite{neves18,Platzer10} that aims at using principles
and techniques from programming theory (\eg\
compositionality~\cite{goncharov18,neves18}, Hoare
calculi~\cite{Platzer10,suenaga11}, theory of
iteration~\cite{adamek11,elgot75}) to provide formal foundations for
developing computational systems that interact with physical
processes. Cruise controllers are a typical example of this pattern; a
very simple case is given by the hybrid program below.
\begin{equation}\label{eq:cc}\tag{cruise controller}
\begin{aligned}
  \!&~\progwhile{true}{\\*~&\qquad\progife{v \leq 10}{(v'= 1 \> \until  \> 1)}
    {(v'= -1 \> \until \> 1)}\\*\!&}  
\end{aligned}
\end{equation}
In a nutshell, the program specifies a digital controller that
periodically measures and regulates a vehicle's velocity ($\prog{v}$):
if the latter is less or equal than $\prog{10}$ the controller
accelerates during $\prog{1}$ time unit, as dictated by the program
statement $\prog{v'= 1 \> \until \> 1}$ ($\prog{v' = 1}$ is a
differential equation representing the velocity's rate of change over
time. The value $\prog{1}$ on the right-hand side of
$\prog{\blue{for}}$ is the duration during which the program statement
runs). Otherwise, it decelerates during the same amount of time
$\prog{(v'= {-1} \> \until \> 1)}$. \Cref{fig:plots} shows the output
respective to this hybrid program for an initial velocity of~5.

\begin{wrapfigure}[12]{r}{68mm}
  \centering
  \vspace{-0.7cm}
  \includegraphics[width=68mm]{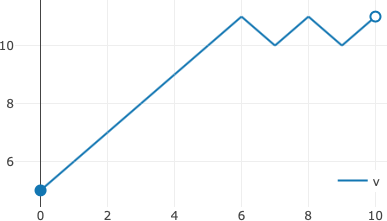}
  \caption{Vehicle's velocity}
  \label{fig:plots}
\end{wrapfigure}
Note that in contrast to standard programming, the cruise controller
involves not only classical constructs (while-loops and conditional
statements) but also differential ones (which are used for describing
physical processes).  This cross-disciplinary combination is the core
feature of hybrid programming and has a notably wide range of
application domains (see~\cite{Platzer10,rajkumar10}).  However, it
also hinders the use of classical techniques of programming, and thus
calls for a principled extension of programming theory to the hybrid
setting.

As is already apparent from the \eqref{eq:cc} example, we stick to an \emph{imperative}
programming style, in particular, in order to keep in touch with the established
denotational models of physical time and computation. A popular alternative to this
for modelling real-time and hybrid systems is to use a \emph{declarative} programming style,
which is done e.g.\ in real-time Maude~\cite{olveczky07} or Modelica~\cite{fritzson2014}.
A well-known benefit of declarative programming is that programs are very easy to 
write, however on the flip side, it is considerably more difficult to define what they 
exactly mean.
 

\myparagraph{Motivation and related work} Most of the previous
research on formal hybrid system modelling has been inspired by
automata theory and Kleene algebra (as the corresponding algebraic
counterpart). These approaches led to the well-known notion of hybrid
automaton~\cite{hybridautomata} and Kleene algebra based languages for
hybrid systems~\cite{Platzer08,algebra_hybrid,HuertayMuniveStruth18}.
From the purely semantic perspective, these formalizations are rather
close and share such characteristic features as \emph{nondeterminism}
and what can be called \emph{non-refined divergence}. The former is
standardly justified by the focus on formal verification of
safety-critical systems: in such contexts overabstraction is usually
desirable and useful. However, coalescing \emph{purely hybrid}
behaviour with nondeterminism detaches semantic models from their
prototypes as they exist in the wild. This brings up several
issues. Most obviously, a nondeterministic semantics, especially not
given in an operational form, cannot directly serve as a basis for
languages and tools for hybrid system testing and simulation.
Moreover, models with nondeterminism baked in do not provide a clear
indication of how to combine hybrid behaviour with effects other than
nondeterminism (e.g.\ probability), or to combine it with
nondeterminism in a different way (\emph{van Glabbeek's
  spectrum}~\cite{Glabbeek90} gives an idea about the diversity of
potentially arising options). Finally, the Kleene algebra paradigm
strongly suggests a relational semantics for programs, with the
underlying relations connecting a state on which the program is run
with the states that the program can reach.  As previously indicated
by H\"{o}fner and M\"{o}ller~\cite{algebra_hybrid}, this view is too
coarse-grained and contrasts to the trajectory-based one where a
program is associated with a trajectory of states (recall
Figure~\ref{fig:plots}). The trajectory-based approach provides an
appropriate abstraction for such aspects as notions of convergence,
periodic orbits, and duration-based
predicates~\cite{duration_calculus1}. This potentially enables analysis
of
properties such as \emph{how fast} our \eqref{eq:cc} example reaches
the target velocity or for \emph{how long} it exceeds~it.

The issue of \emph{non-refined divergence} mentioned earlier arises
from the Kleene algebra law ${\prog{p \scomp 0} = \prog{0}}$
in conjunction
with Fischer-Ladner's encoding of while-loops
$\progwhile{b}{p}$ as $(\prog{b \scomp p})^\ast;\neg
\prog{b}$. 
This creates a havoc with all divergent programs
\begin{flalign*}
 \progwhile{true}{p} 
\end{flalign*}
as they become identified with divergence~$\prog{0}$, thus making the
above example of a \eqref{eq:cc} meaningless.  This issue is
extensively discussed in H\"{o}fner and M\"{o}ller's
work~\cite{algebra_hybrid} on a \emph{nondeterministic} algebra of
trajectories, which tackles the problem by disabling the law
$\prog{p \scomp 0} = \prog{0}$ and by introducing a special operator
for infinite iteration that inherently relies on nondeterminism. This
iteration operator inflates trajectories at so-called `Zeno points'
with arbitrary values, which in our case would entail \eg\ the program
\begin{flalign}\label{eq:zeno}\tag{zeno}
 \prog{x \ass 1} \scomp \progwhile{true}{\blue{wait} \> x \scomp x\ass
  x/2} 
\end{flalign}
to output at time instant $2$ all possible values in the valuation
space (the expression $\prog{\blue{wait} \> t}$ represents a wait call
of $\prog{t}$ time units).  More details about Zeno points can be
consulted in~\cite{algebra_hybrid,goncharov2019}.

In previous work~\cite{goncharov18,goncharov2019}, we pursued a
\emph{purely hybrid} semantics via a simple \emph{deterministic
  functional} language \hc{}, with while-loops for which we used
Elgot's notion of iteration~\cite{elgot75} as the underlying semantic
structure. That resulted in a semantics of finite and infinite
iteration, corresponding to a refined view of divergence.
Specifically, we developed an operational
semantics and also a denotational counterpart for \hc{}. An important
problem of that semantics, however, is that it involves infinitely
many premisses and requires calculating total duration of programs,
which precludes using such semantics directly in implementations.
Both the above examples~\eqref{eq:cc} and~\eqref{eq:zeno} are affected
by this issue. In the present paper we propose an \emph{imperative}
language with a denotational semantics similar to \hc{}'s one, but now
provide a clear recipe for executing the semantics in a constructive
manner.%

\myparagraph{Overview and contributions} Building on our previous
work~\cite{goncharov2019}, we devise operational and denotational
semantics suitable for implementation purposes, and provide a
soundness and adequacy theorem relating both these styles of
semantics.  Results of this kind are well-established yardsticks in
the programming language theory~\cite{winskel93}, and are also
beneficial from a practical perspective.  For example, small-step
operational semantics naturally guides the implementation of
compilers/evaluators for programming languages, whilst denotational
semantics is more abstract, syntax-independent, and guides the study
of program equivalence, of the underlying computational paradigm, and
its combination with other computational effects.

As mentioned before, in our previous work~\cite{goncharov2019} we
introduced a simple functional hybrid language~\hc{} with operational
and denotational monad-based semantics. Here, we work with a similar
imperative while-language, whose semantics is given in terms of a
global state space of trajectories over $\Reals^n$, which is a
commonly used carrier when working with solutions of systems of
differential equations. A key principle we have taken as a basis for
our new semantics is the capacity to determine behaviours of a
program~$\prog{p}$ by being able to examine only some subterms of it.
In order to illustrate this aspect, first note that our semantics does
not reduce program terms $\prog{p}$ and initial states~$\sigma$
(corresponding to valuation functions $\sigma\c \Vars \to \mathbb{R}$
on program variables $\Vars$) to states $\sigma'$, as usual in
classical programming. Instead it reduces \emph{triples}
$\tpl{\prog{p}\sep\sigma\sep t}$ of programs $\prog{p}$, initial
states $\sigma$ and time instants $\prog{t}$ to a state $\sigma'$;
such a reduction can be read as ``given~$\sigma$ as the initial state,
program~$\prog{p}$ produces a state $\sigma'$ at time instant
$\prog{t}$''. Then, the reduction process of
${\tpl{\prog{p}\sep \sigma\sep t}}$ to a state only examines fragments
of $\prog{p}$ or unfolds it when strictly necessary, depending of the
time instant~$\prog{t}$. For example, the reduction of the
\eqref{eq:cc} unfolds the underlying loop only twice for the time
instant $1 + \nicefrac{1}{2}$ (the time instant $1 + \nicefrac{1}{2}$
occurred in the second iteration of the loop).  This is directly
reflected in our prototype implementation of an interactive evaluator
of hybrid programs~\Lince. It is available online and comes with a
series of examples for the reader to explore
(\url{http://arcatools.org/lince}). The plot in \Cref{fig:plots} was
automatically obtained from \Lince, by calling on the previously
described reduction process for a predetermined sequence of time
instants $\prog{t}$.

For the denotational model, we build on our previous
work~\cite{goncharov18,goncharov2019} where hybrid programs are
interpreted via a suitable monad $\BBH$, called the~\emph{hybrid
  monad} and capturing the computational effect of \emph{hybridness},
following the seminal approach of
Moggi~\cite{moggi:1989,moggi:1991}. Our present semantics is more
lightweight and is naturally couched in terms of another
monad~$\BBH_S$, parametrized by a set $S$. In our case, as
mentioned above, $S$ is the set of trajectories over $\Reals^n$ where $n$
is the number of available program variables $\Vars$. The latter monad
is in fact parametrized in a formal sense~\cite{uustalu03} and comes
out as an instance of a recently emerged generic
construction~\cite{diezel20}. A remarkable salient feature of that
construction is that it can be instantiated in a constructive setting
(without using any choice principles) -- although we do not touch upon
this aspect here, in our view this reinforces the fundamental nature
of our semantics. Among various benefits of~$\BBH_S$ over~$\BBH$, the
former monad enjoys a construction of an iteration operator (in the
sense of Elgot~\cite{elgot75}) as a \emph{least fixpoint}, calculated
as a limit of an $\omega$-chain of approximations, while for $\BBH$
the construction of the iteration operator is rather intricate and no
similar characterization is available. A natural question that arises
is: how are $\BBH$ and $\BBH_S$ related? We do answer it by providing
an instructive connection, which sheds light on the construction
of~$\BBH$, by explicitly identifying semantic ingredients which have
to be added to $\BBH_S$ to obtain~$\BBH$. Additionally, this results
in ``backward compatibility'' with our previous work.

\myparagraph{Document structure} After short preliminaries
(\Cref{sec:prelim}), in \Cref{sec:syn_sem} we introduce our
while-language and its operational semantics.
In~\Cref{sec:monad,sec:monads}, we develop the denotational model for our
language and connect it formally to the existing hybrid
monad~\cite{goncharov18,goncharov2019}.
In~\Cref{sec:deno}, we prove a soundness and adequacy result for our operational 
semantics w.r.t.\ the developed model.
\Cref{sec:arch} describes \Lince's architecture.
Finally, \Cref{sec:concl} concludes and briefly discusses future
work. Omitted proofs are found in appendix for reviewing purposes.

\section{Preliminaries}\label{sec:prelim}
We assume familiarity with category
theory~\cite{AdamekHerrlichEtAl90}.  By $\real$, $\realp$ and
$\realpe$ we respectively denote the sets of reals, non-negative
reals, and extended non-negative reals (i.e.\,$\realp$ extended with
the infinity value $\infty$).
Let $[0,\realpe\rrparenthesis$ denote the set of downsets of $\realpe$
having the form~$[0,d]$ ($d\in\realp$) or the form~$[0,d)$
($d\in\realpe$). We call the elements of the dependent sum
$\sum_{I\in [0,\realpe\rrparenthesis} X^I$ \emph{trajectories}
(over~$X$).
By $[0,\realp]$, $[0,\realp)$ and $[0,\realpe)$ we denote the
following corresponding subsets of $[0,\realpe\rrparenthesis$:
$\{[0,d]\mid d\in\realp\}$, $\{[0,d)\mid d\in\realp\}$ and
$\{[0,d)\mid d\in\realpe\}$.
%
%
%
By $X\uplus Y$ we denote the \emph{disjoint union}, which is the
categorical coproduct in the category of sets with the corresponding
left and right injections $\inl\c X\to X\uplus Y$,
$\inr\c Y\to X\uplus Y$. To reduce clutter, we often use plain union
$X\cup Y$ in place of $X\uplus Y$ if $X$ and $Y$ are disjoint by
construction.

By $\ite{a}{b}{c}$ we denote the case distinction construct: $a$ if
$b$ is true and $c$ otherwise. By~$\bang$ we denote the \emph{empty
  function}, i.e.\ a function with the empty domain. For the sake of
succinctness, we use the notation $e^t$ for the function application
$e(t)$ with real-value $t$.

\section{An imperative hybrid while-language and its
  semantics} \label{sec:syn_sem}

This section introduces the syntax and operational semantics of our
language.  We first fix a stock of $n$-variables
$\Vars = \{ \prog{x_1}, \dots, \prog{x_n}\}$ over which we build
atomic programs, according to the grammar
\begin{align*}
  \prog{At(\Vars)} ~\ni~& \prog{x\ass t} ~\mid ~
    \prog{x'_1 = t_1}, \dots, \prog{x'_n = t_n} \> \prog{\until} \>
    \prog{t} \quad \>\\[1ex]
  \prog{LTerm(\Vars)} ~\ni~& \prog{r} ~\mid~ \prog{r} \cdot \prog{x} ~\mid~
    \prog{t + s} 
\end{align*}
where $\prog{x} \in \Vars $, $\prog{r} \in \Reals$, $\prog{t_i}, \prog{t},\prog{s} \in \prog{LTerm(\Vars)}$. 
 An atomic program is thus either a classical assignment
$\prog{x\ass t}$ or a differential statement
$\prog{x'_1 = t_1}, \dots, \prog{x'_n = t_n} \> \prog{\until} \>
\prog{t}$. The latter reads as ``\emph{run the system of differential
  equations $\prog{x'_1 = t_1}, \dots, \prog{x'_n = t_n}$ for $\prog{t}$
  time units}''. We then define the while-language via the grammar
\begin{align*}
  &\prog{Prog(\Vars)} \ni \prog{a}
    ~\mid~ \prog{p \scomp q} ~\mid~
    \progife{b}{p}{q} ~\mid~
    \progwhile{b}{p}
\end{align*}
where $\prog{p,q} \in \prog{Prog(\Vars)}$,
$\prog{a} \in \prog{At(\Vars)}$ and $\prog{b}$ is an element of the
free Boolean algebra generated by the terms $\prog{t} \leq \prog{s}$
and $\prog{t} \geq \prog{s}$. The expression $\prog{\blue{wait} \> t}$
(from the previous section) is encoded as the differential statement
$\prog{x'_1 = 0, \dots, x'_n = 0} \> \prog{\until} \> \prog{t}$.

\begin{remark}
  The systems of differential equations that our language allows are
  always linear. This is not to say that we could not consider more
  expressive systems; in fact we could straightfowardly extend the
  language in this direction, for its semantics (presented below) is
  not impacted by specific choices of solvable systems of differential
  equations. But here we do not focus on such choices regarding the
  expressivity of continuous dynamics and concentrate on a core hybrid
  semantics instead on which to study the fundamentals of hybrid
  programming.
\end{remark}

In the sequel we abbreviate differential statements
$\prog{x'_1 = t_1, \dots, x'_n = t_n} \> \prog{\until} \> \prog{t}$ to
the expression $\prog{\bar{x}' = \bar{t}} ~ \prog{\until} ~ \prog{t}$,
where $\prog{\bar{x}}'$ and $\prog{\bar{t}}$ abbreviate the
corresponding vectors of variables $\prog{x_1}' \dots \prog{x_n}'$ and
linear-combination terms $\prog{t_1} \dots \prog{t_n}$. We call
functions of type $\sigma\c \Vars \to \Reals$ \emph{environments};
they map variables to the respective valuations.  We use the
notation $\sigma\triangledown [\prog{\bar v}/ \prog{\bar x}]$ to
denote the environment that maps each $\prog{x_i}$ in $\prog{\bar x}$
to $\prog{v_i}$ in $\prog{\bar{v}}$ and the rest of variables in the
same way as $\sigma$.  Finally, we denote by
$\phi_\sigma^{\bar{\prog{x}}' = \bar{\prog{t}}} \c [ 0,\infty) \to
\Reals^n$ the solution of a system of differential equations
$\bar{\prog{x}}' = \bar{\prog{t}}$ with $\sigma$ determining the
initial condition. When clear from context, we omit the superscript in
$\phi_\sigma^{\bar{\prog{x}} = \bar{\prog{t}}}$. For a
linear-combination term $\prog{t}$ the expression $\prog{t}\sigma$
denotes the corresponding interpretation according to $\sigma$ and
analogously for $\prog{b}\sigma$ where $\prog{b}$ is a Boolean
expression.

We now introduce a small-step operational semantics for our
language. Intuitively, the semantics establishes a set of rules for
reducing a triple $\langle$program statement, environment, time
instant$\rangle$ to an environment, via a \emph{finite} sequence of
reduction steps. The rules are presented in \Cref{small_step}.  The
terminal configuration $\langle\skp, \sigma, \prog{t}\rangle$ represents a
successful end of a computation, which can then be fed into another
computation (via rule \textbf{(seq-skip$^\to$)}).  Contrastingly,
$\langle\stp, \sigma, \prog{t}\rangle$ is a terminating configuration that
inhibits the execution of subsequent computations. The latter is
reflected in rules \textbf{(diff-stop$^\to$)} and
\textbf{(seq-stop$^\to$)} which entail that, depending on the chosen
time instant, we do not need to evaluate the whole program, but merely
a part of it
-- consequently,
infinite while-loops need not yield infinite reduction
sequences (as explained in
Remark~\ref{rem:inf}). 
Note that time $\prog{t}$ is consumed when applying the rules
\textbf{(diff-stop$^\to$)} and \textbf{(diff-seq$^\to$)} in
correspondence to the duration of the differential statement at hand. The
rules \textbf{(seq)} and \textbf{(seq-skip$^\to$)} correspond to the
standard rules of operational semantics for while languages over an
imperative store~\cite{winskel93}.
\newcommand{\prem}[1]{(\textit{if\/ }#1)}
\newcommand{\nline}{\vspace{-8mm}}
\begin{figure*}[t]
\begin{minipage}{1\textwidth}
\begin{flalign*}
\textbf{(asg$^\to$)}
&&
\prog{x \ass t}\sep\sigma\sep t 
    \ssto
  \skp\sep\sigma\triangledown [\prog{t}\sigma/\prog{x}]\sep t 
&&
\end{flalign*} \nline
\begin{flalign*}
\textbf{(diff-stop$^\to$)}
  &&\prog{\bar{x}' = \bar{u} \>
  \until \> \prog{t}} \sep\sigma\sep t 
    \ssto
  \stp\sep\sigma\triangledown[\phi_\sigma(t)/\bar{\prog{x}}]\sep 0
  &&
  \prem{t < \prog{t}\sigma}
\end{flalign*} \nline
\begin{flalign*}
\textbf{(diff-skip$^\to$)}
&& \prog{\bar{x}' = \bar{u} \>
  \until \> \prog{t}} \sep\sigma\sep t 
    \ssto
  \skp\sep\sigma\triangledown[\phi_\sigma(\prog{t}\sigma)/\bar{\prog{x}}]\sep t - (\prog{t}\sigma)
&&
\prem{t\geq \prog{t}\sigma}
\end{flalign*} \nline
\begin{flalign*}
\textbf{(if-true$^\to$)}
&&
\progife{b}{p}{q}\sep\sigma\sep t \ssto \prog{p} \sep\sigma\sep t
&&
\prem{\prog{b}\sigma=\top}
\end{flalign*} \nline
\begin{flalign*}
\textbf{(if-false$^\to$)}
&&
\progife{b}{p}{q}\sep\sigma\sep t \ssto \prog{q}\sep\sigma\sep t
&&
\prem{\prog{b}\sigma=\bot}
\end{flalign*} \nline
\begin{flalign*}
\textbf{(wh-true$^\to$)}
&&
\progwhile{b}{p}\sep\sigma\sep t \ssto \prog{p} \scomp \progwhile{b}{p}\sep\sigma\sep t
&&
\prem{\prog{b}\sigma=\top}
\end{flalign*} \nline
\begin{flalign*}
\textbf{(wh-false$^\to$)}
&&
\progwhile{b}{p}\sep\sigma\sep t \ssto \skp\sep\sigma\sep t
&&
\prem{\prog{b}\sigma=\bot}
\end{flalign*} \vspace{-6mm}
\begin{flalign*}
\lrule{(seq-stop$^\to$)}{\prog{p}\sep\sigma\sep t \ssto \stp\sep\sigma'\sep t'}{
  \prog{p}\scomp \prog{q}\sep\sigma\sep t  \ssto \stp\sep\sigma'\sep t'}
&&
\lrule{(seq-skip$^\to$)}{\prog{p}\sep\sigma\sep t \ssto \skp\sep\sigma'\sep t'}{
  \prog{p}\scomp \prog{q}\sep\sigma\sep t  \ssto \prog{q}\sep\sigma'\sep t'}
\end{flalign*} \vspace{-5mm}
\begin{flalign*}
&&\lrule{(seq$^\to$)}{\prog{p}\sep\sigma\sep t \ssto \prog{p'}\sep\sigma'\sep t'}{
  \prog{p}\scomp \prog{q}\sep\sigma\sep t  \ssto \prog{p'};\prog{q}\sep\sigma'\sep t'} \qquad \prem{\prog{p'} \neq \stp \textit{ and } \prog{p'} \neq \skp} &&
\end{flalign*} 
  \end{minipage}
  \caption{Small-step Operational Semantics}
  \label{small_step}
\end{figure*}
\begin{remark} \label{rem:inf} Putatively infinite while-loops do not necessarily
  yield infinite reduction steps. Take for example the while-loop
  below whose iterations have always  duration $\prog{1}$.
  \begin{flalign} \label{inf:loop}
    \prog{x\ass 0} \scomp \progwhile{true}{x\ass x+1 \> \scomp \> \prog{\blue{wait}
        \> 1}}
  \end{flalign}
  It yields a finite reduction sequence for the time instant
  $\nicefrac{1}{2}$, as shown below:
  \begin{align*}
    & \prog{x\ass 0} \scomp \progwhile{true}{x\ass x+1 \> \scomp \>
      \prog{\blue{wait}
        \> 1}}\sep \sigma \sep \nicefrac{1}{2}  \to \\
    & \qquad\{ \text{by the rules \textbf{(asg$^\to$)} and \textbf{(seq-skip$^\to$)}} \} \\
    & \progwhile{true}{x\ass x+1 \> \scomp \>
      \prog{\blue{wait}
        \> 1}}\sep \sigma\triangledown [\prog{0}/ \prog{x}]\sep \nicefrac{1}{2}  \to \\
    & \qquad\{ \text{by the rule \textbf{(wh-true$^\to$)}} \} \\
    & \prog{x\ass x+1} \> \scomp \> \prog{\blue{wait} \> 1} \scomp
    \progwhile{true}{x\ass x+1 \> \scomp \> \prog{\blue{wait}
        \> 1}}\sep \sigma\triangledown [\prog{0}/ \prog{x}] \sep \nicefrac{1}{2} \to \\
    & \qquad\{ \text{by the rules \textbf{(asg$^\to$)} and \textbf{(seq-skip$^\to$)}} \} \\
    & \prog{\blue{wait} \> 1} \scomp \progwhile{true}{x\ass x+1
      \> \scomp \> \prog{\blue{wait} \> 1}}\sep \sigma\triangledown [\prog{0 + 1}/ \prog{x}]\sep \nicefrac{1}{2} \to \\
    & \qquad\{ \text{by the rules \textbf{(diff-stop$^\to$)} and \textbf{(seq-stop$^\to$)}} \} \\
    & \stp\sep \sigma\triangledown [\prog{0 + 1}/ \prog{x}]\sep 0 
  \end{align*}
  The gist is that to evaluate program~\eqref{inf:loop} at time
  instant $\nicefrac{1}{2}$, one only needs to unfold the underlying
  loop until surpassing $\nicefrac{1}{2}$ in terms of execution
  time. Note that if the wait statement is removed from the program
  then the reduction sequence would not terminate, intuitively because
  all iterations would be instantaneous and thus the total execution
  time of the program would never reach $\nicefrac{1}{2}$.
\end{remark}
The following theorem entails that our semantics
is deterministic, which is instrumental for our implementation.

\begin{theorem} \label{thm:determ}
  For every program $\prog{p}$, environment $\sigma$, and time instant
  $t$ there is \emph{at most one} applicable reduction rule.
\end{theorem}
Let $\ssto[\star]$ be the transitive closure of the reduction relation
$\ssto$ that was previously presented.

\begin{corollary}
  For every program term $\prog{p}$, environments $\sigma$, $\sigma'$,
  $\sigma''$, time instants $\prog{t}$, $\prog{t}'$, $\prog{t}''$, and termination flags
  $\prog{s},\prog{s'} \in \{\skp,\stp\}$, if
  $\prog{p}\sep\sigma\sep t \ssto[\star] s, \sigma', t'$ and
  $\prog{p}\sep\sigma\sep t \ssto[\star] s'\sep\sigma''\sep t''$, then
  the equations $\prog{s} = \prog{s'}$, $\sigma' = \sigma''$ and $\prog{t}' = \prog{t}''$ must
  hold.
\end{corollary}
\begin{proof}
  Follows by induction on the number of reduction steps and Theorem\,\ref{thm:determ}.%
\qed\end{proof}
As alluded above, the operational semantics treats time as a
resource. This is formalised below.
\begin{proposition}\label{prop:shift}
  For all program terms \prog{p} and \prog{q}, environments $\sigma$
  and $\sigma'$, and time instants $\prog{t}$, $\prog{t'}$ and $\prog{s}$, if\/
  $\prog{p}\sep\sigma\sep t \ssto \prog{q}\sep\sigma'\sep \prog{t'}$ then\/
  $\prog{p}\sep\sigma\sep t + s \ssto \prog{q}\sep\sigma'\sep t' + s$;
  and if\/ $\prog{p}\sep\sigma\sep t \ssto \skp\sep\sigma'\sep t'$
  then\/
  $\prog{p}\sep\sigma\sep t+s \ssto \skp\sep\sigma'\sep {t'+s}$.
\end{proposition}
%

\section{Towards Denotational Semantics: The Hybrid Monad}
\label{sec:monad}
%
A mainstream subsuming paradigm in denotational semantics is due to
Moggi \cite{moggi:1989,moggi:1991}, who proposed to identify a
\emph{computational effect} of interest as a monad, around which the
denotational semantics is built using standard generic mechanisms,
prominently provided by category theory. In this section we recall 
necessary notions and results, motivated by this approach, to 
prepare ground for our main constructions in the next section.
\begin{definition}[Monad]
  A monad $\BBT$ (on the category of sets and functions) is given by a
  triple $(T,\eta, (\argument)^\star)$, consisting of an endomap~$T$
  over the class of all sets, together with a set-indexed class of
  maps $\eta_X\c X\to TX$ and a so-called \emph{Kleisli lifting}
  sending each $f\c X \to T Y$ to $f^\star\c TX \to TY$ and
obeying \emph{monad laws}: $\eta^{\star}=\id$, 
$f^{\star}\comp\eta=f$,  
$(f^{\star}\comp g)^{\star}=f^{\star}\comp g^{\star}$ (it follows from this definition
that~$T$ extends to a functor and $\eta$ to a natural transformation).

A monad morphism $\theta\c\BBT\to\BBS$ from
$(T,\eta^\BBT, (\argument)^{\star\BBT})$ to
$(S,\eta^\BBS, (\argument)^{\star\BBS})$ is a natural
transformation $\theta\c T\to S$ such  that
$\theta\comp\eta^\BBT = \eta^\BBS$ and
$\theta\comp f^{\star\BBT} = (\theta\comp
f)^{\klstar\BBS}\comp\theta$.
\end{definition}
We will continue to use bold capitals (e.g.~$\BBT$) for monads over the corresponding
endofunctors written as capital Romans (e.g.~$T$).

In order to interpret while-loops one needs additional structure on
the monad.
\begin{definition}[Elgot Monad]\label{def:elgot}
  A monad $\BBT$ 
  is called \emph{Elgot} if it is equipped with an \emph{iteration}
  operator $(\argument)^\dagger$ that sends each $f\c X \to T(Y \uplus X)$
  to $f^\dagger\c X \to T Y$ in such a way that certain established axioms of 
  iteration are satisfied~\cite{adamek11,goncharov17}. 
  
  Monad morphisms between Elgot monads are additionally required to preserve iteration: 
$\theta\comp f^{\istar\BBT} = (\theta\comp f)^{\istar\BBS}$ for $\theta\c\BBT\to\BBS$, $f\c X\to T(Y\uplus X)$.
\end{definition}
%
For a monad~$\BBT$, a map $f\c X\to TY$, called a \emph{Kleisli map},
is roughly to be regarded as a semantics of a program $\prog{p}$, with
$X$ as the semantics of the input, and $Y$ as the semantics of the
output.  For example, with $T$ being the \emph{maybe monad}
$(\argument)\uplus\{\bot\}$, we obtain semantics of programs as
partial functions. Let us record this example in more detail for
further reference.
\begin{example}[Maybe Monad $\BBM$]\label{exp:maybe}
The maybe monad is determined by the following data: 
$MX=X\uplus\{\bot\}$, the unit is the left injection $\inl\c X\to X\uplus\{\bot\}$
and given $f\c X\to Y\uplus\{\bot\}$, $f^\klstar$ is equal to the
copairing $[f,\inr]\c X\uplus\{\bot\}\to Y\uplus\{\bot\}$.

It follows by general considerations (enrichment of the category of
Kleisli maps over complete partial orders) that $\BBM$ is an Elgot
monad with the following iteration operator~$(\argument)^\natural$:
given $f\c X\to (Y\uplus X)\uplus\{\bot\}$, and $x_0\in X$, let
$x_0,x_1,\ldots$ be the longest (finite or infinite) sequence over $X$
constructed inductively in such a way that
$f(x_i) = \inl\inr x_{i+1}$. Now, $f^\natural(x_0) = \inr\bot$ if the
sequence is infinite or $f(x_i) = \inr\bot$ for some $i$, and
$f^\natural(x_0)=\inl y$ if for the last element of the sequence
$x_n$, which must exist, $f(x_n) = \inl\inl y$.

Other examples of Elgot monad can be consulted \eg\
in~\cite{goncharov17}.
\end{example}
The computational effect of \emph{hybridness} can also be captured by
a monad, called \emph{hybrid monad}~\cite{goncharov18,goncharov2019},
which we recall next (in a slightly different but equivalent form). To
that end, we also need to recall \emph{Minkowski addition} for subsets
of the set $\realpe$ of extended non-negative reals (see
Section~\ref{sec:prelim}): $A+B = \{a+b\mid a\in A\comma b\in B\}$,
e.g.\ $[a,b] + [c,d] = [a + c, b + d]$ and
$[a,b] + [c,d) = [a + c, b + d)$.
\begin{definition}[Hybrid Monad $\BBH$]
The \emph{hybrid monad} $\BBH$ is defined as follows.
\begin{itemize}
\item $H X = \sum_{I\in[0,\realp]} X^I \uplus  \sum_{I\in[0,\realpe\rrparenthesis} X^I$, \ie\
  it is a set of trajectories valued on $X$ and with the domain downclosed.
  For any $p =\inj \brks{I,e}\in HX$ with $\inj\in\{\inl\comma\inr\}$,
  let us use the notation $p_\dr = I$, $p_\ev = e$, the former being the duration
  of the trajectory and the latter the trajectory itself. Let also $\eps=\brks{\emptyset,\bang}$. 
\item $\eta(x) = \inl\brks{[0,0],\l t.\,x}$, \ie\ $\eta(x)$ is a
  trajectory of duration $0$ that returns $x$.
\item given $f\colon X\to HY$, we define $f^\klstar\colon HX\to HY$
  via the following clauses:
\begin{flalign*}
  \quad f^\klstar (\inl\brks{I,e}) =&\; \inj\brks{I+J,\l t.\,\ite{(f(e^t))_\ev^0}{t<d}{(f(e^d))_\ev^{t-d}}}\\
  &&\hspace{-15em} \text{if~~$I'=I=[0,d] \text{ for some } d, f(e^d) = \inj\,\brks{J,e'}$} \\
%
%
\quad f^\klstar (\inl\brks{I,e}) =&\; \inr\brks{I',\l t.\, (f(e^t))_\ev^0}
&&\hspace{-4em} \text{if~~$I'\neq I$}\\
\quad f^\klstar (\inr\brks{I,e}) =&\; \inr\brks{I',\l t.\, (f(e^t))_\ev^0}
\end{flalign*}
where $I'=\bigcup\,\bigl\{[0,t]\subseteq I \mid \forall s\in [0,t].\, f(e^s)\neq\inr\eps\bigr\}$
and $\inj\in\{\inl\comma\inr\}$.
\end{itemize}
\end{definition}
The definition of the hybrid monad $\BBH$ is somewhat intricate, so
let us complement it with some explanations (details and further
intuitions about the hybrid monad can also be consulted
in~\cite{goncharov18}).  The domain $HX$ constitutes three types of
trajectories representing different kinds of hybrid computation:
\begin{itemize}
  \item\emph{(closed) convergent}: $\inl\brks{[0,d],e}\in HX$ (e.g.\ instant termination $\eta(x)$);
  \item\emph{open divergent}: $\inr\brks{[0,d),e}\in HX$ (e.g.\
    instant divergence $\inr\eps$ or a trajectory $[0,\infty) \to X$
    which represents a computation that runs \emph{ad infinitum});
  \item\emph{closed divergent}: $\inr\brks{[0,d],e}\in HX$
    (representing computations that start to diverge \emph{precisely}
    after the time instant $d$).
\end{itemize}
The Kleisli lifting $f^\klstar$ works as follows: for a given
trajectory $\inj \brks{I,e}$, we first calculate the largest interval
$I'\subseteq I$ on which the trajectory $\lambda t \in I'. f(e^t)$
does not instantly diverge (\ie $f(e^t) \neq \inr \eps$) throughout,
hence $I'$ is either $[0,d']$ or $[0,d')$ for some $d'$.  Now, the
first clause in the definition of~$f^\star$ corresponds to the
successful composition scenario: the argument trajectory $\brks{I,e}$
is convergent, and composing~$f$ with $e$ as described in the
definition of $I'$ does not yield divergence all over $I$. In that
case, we essentially concatenate $\brks{I,e}$ with $f(e^d)$, the
latter being the trajectory computed by $f$ at the last point of~$e$.
The remaining two clauses correspond to various flavours of
divergence, including divergence of the input~($\inr\brks{I,e}$) and
divergences occurring along $f\comp e$. Incidentally, this explains
how closed divergent trajectories may arise: if $I'=[0,d']$ and $d'$
is properly smaller than $d$, this indicates that we diverge precisely
\emph{after} $d'$, which is possible e.g.\ if the program behind $f$
continuously checks a condition which did not fail up until $d'$.

\section{Deconstructing the Hybrid Monad}\label{sec:monads}

As mentioned in the introduction, in \cite{goncharov2019} we used
$\BBH$ for giving semantics to a \emph{functional} language~\hc{}
whose programs are interpreted as morphisms of type $X\to
HY$. Here, we are dealing with an \emph{imperative} language, which
from a semantic point of view amounts to fixing a type of
\emph{states} $S$, shared between all programs; the semantics
of a program is thus restricted to morphisms of type~$S\to HS$. As explained next, 
this allows us to make do with a simpler monad~$\BBH_S$, globally parametrized 
by~$S$. The new monad $\BBH_S$ has the property that $H_S S$ 
is naturally isomorphic to~$HS$. Apart from (relative to $\BBH$) simplicity, the 
new monad enjoys further benefits, specifically $\BBH_S$ 
is mathematically a better behaved structure, e.g.\ in contrast to~$\BBH$, 
Elgot iteration on $\BBH_S$ is constructed as a least fixed
point. Factoring the denotational semantics through $\BBH_S$ thus
allows us to bridge the gap to the operational semantics given in
Section~\ref{sec:syn_sem}, and faciliates the soundness and
adequacy proof in the forthcoming Section~\ref{sec:deno}.  

In order to
define $\BBH_S$, it is convenient to take a slightly broader
perspective. We will also need to make a detour through the topic of
ordered monoid modules with certain completeness properties so that we
can characterise iteration on $\BBH_S$ as a least fixed point.

\begin{definition}[Monoid Module, Generalized Writer Monad~\cite{goncharov2019}]\label{def:gen-write}
  Given a (not necessarily commutative) monoid $(\om, +,0)$, a \emph{monoid
    module} is a set $\mm$ equipped with a map
  $\later\c\om\times\mm\to\mm$ (monoid action), subject to the laws 
  $0\later e = e$, $(m + n)\later e = m\later (n\later e)$.

  Every monoid-module
  pair $(\om,\mm)$ induces a \emph{generalized
    writer monad} $\BBT=(T,\eta,(\argument)^\klstar)$ with
  $T=\om\times (\argument)\cup \mm$, $\eta_X(x) = \brks{0,x}$, and
\begin{align*}
\qquad f^\star(m,x) =\;& (m + n,y) &&\text{\quad where\qquad $m\in\om$, $x\in X$, $f(x) = \brks{n,y}\in\om\times Y$}&&\\*
f^\star(m,x) =\;& m\later e &&\text{\quad where\qquad $m\in\om$, $x\in X$, $f(x) = e\in\mm$}&&\\*
f^\star(e) =\;& e &&\text{\quad where\qquad $e\in\mm$}
\end{align*}
This generalizes the writer mo\-nad $(\mm=\emptyset)$
and the exception mo\-nad~${(\om=1)}$. 
\end{definition}
\begin{example}
  A simple motivating example of a monoid-module pair $(\om,\mm)$ is the pair
  $(\realp,\realpe)$ where the monoid operation is addition with $0$
  as the unit and the monoid action is also addition.
\end{example}

More specifically, we are interested in \emph{ordered monoids} and 
\emph{(conservatively) complete monoid modules}. These are defined as follows.
\begin{definition}[Ordered Monoids, (Conservatively) Complete Monoid Modules~\cite{diezel20}]
We call a monoid $(\om,0,+)$ an \emph{ordered monoid} if it is equipped with a partial 
order $\leq$, such that $0$ is the least element of this order and $+$ is right-monotone
(but not necessarily left-monotone).

An \emph{ordered $\om$-module} w.r.t.\ an ordered monoid $(\om, +, 0, \le)$, is an $\om$-module $(\mm,\later)$ together with a partial order $\appr$ 
and a least element~$\bot$, such that $\later$ is monotone on the right and $(- \later \bot)$ is monotone, i.e.
\begin{align*}
\infrule{}{\bot\appr x}&&
\infrule{x \appr y}{a \later x \appr a \later y}&&
\infrule{a \le b}{a \later \bot \appr b \later \bot}&&
\end{align*}
We call the last property \emph{restricted left monotonicity}.

An ordered $\om$-module is \emph{($\omega$-)complete} 
if for every $\omega$-chain $s_1\appr s_2\appr\ldots$ on~$\mm$ there is a least upper 
bound $\lub_i s_i$ and $\later$ is continuous on the right, i.e.\
\begin{align*}
\infrule{}{\forall i.\, s_i \appr \lub_i s_i}&&
\infrule{\forall i.\,s_i \appr x}{\lub_i s_i \appr x}&&
\infrule{}{a \later \lub_i s_i\appr\lub_i a \later s_i}
\end{align*}
(the law $\lub_i a \later s_i \appr a \later \lub_i s_i$ is derivable). 
Such an $\om$-module is \emph{conservatively complete} if additionally 
for every $\omega$-chain $a_1\appr a_2\appr\ldots$ in $\om$, such that the least upper 
bound $\bigor_i a_i$ exists, $\bigl(\bigor_i a_i\bigr)\later\bot = \lub_i~ a_i\later\bot$. 

A homomorphism $h\c\mm\to\mmm$ of (conservatively) complete monoid
$\om$-modules is required to be monotone and structure-preserving in
the following sense: $h(\bot)=\bot$, $h(a\later x) = a\later h(x)$,
$h(\lub_i x_i) = \lub_i h(x_i)$.
\end{definition}
The completeness requirement for $\om$-modules has a standard
motivation coming from domain theory, where $\appr$ is regarded as
an \emph{information order} and completeness is needed to ensure that
the relevant semantic domain can accommodate infinite behaviours. The
conservativity requirement additionally ensures that the least upper
bounds, which exist in~$\om$ agree with those in $\mm$.  Our main
example is as follows (we will use it for building $\BBH_S$ and its
iteration operator).
\begin{definition}[Monoid Module of Trajectories]\label{defn:mon}
The ordered monoid of \emph{finite open trajectories}                              
$\bigl(\fot_S, \,\tconc\,, \brks{\emptyset,\bang}, \le\bigr)$ over a given  
set $S$, is defined as follows: $\fot_S = \sum_{I\in [0,\realp)} S^I$, the unit is the empty trajectory $\eps = \brks{\emptyset,!}$; summation is
concatenation of trajectories $\,\tconc\,$, defined as follows:                         
\begin{align*}
\brks{[0,d_1),e_1} \tconc \brks{[0,d_2),e_2} =&\; \brks{[0,d_1 + d_2), \lambda t.\, \ite{e_1^t}{t < d_1}{e_2^{t - d_1}}}.
\end{align*}
The relation $\le$ is defined as follows: 
$\brks{[0,d_1),e_1} \le \brks{[0,d_2),e_2}$ if $d_1 \le d_2$ and $e_1^t = e_2^t$ for every $t\in[0,d_1)$. 
%
We can additionally consider both sets $\sum_{I\in [0,\realpe)} S^I$
and $\sum_{I\in [0,\realpe\rrparenthesis} S^I$ as $\fot_S$-modules, by
defining the monoid action $\later$ also as concatenation of
trajectories and by equipping these sets with the order $\sqsubseteq$:
$\brks{I_1,e_1} \sqsubseteq \brks{I_2,e_2}$ if $I_1 \subseteq I_2$ and
$e_1^t = e_2^t$ for all $t \in I_1$.
\end{definition}
Consider the following functors:
\begin{align}
\label{eq:H-def-pre}
H'_SX =&\; \sum_{I\in [0,\realp)} S^I\times X\cup \sum_{I\in [0,\realpe)} S^I\\
\label{eq:H-def}
H_SX  =&\; \sum_{I\in [0,\realp)} S^I\times X\cup \sum_{I\in [0,\realpe\rrparenthesis} S^I
\end{align}
Both of them extend to monads $\BBH'_S$ and $\BBH_S$ as they are
instances of Definition~\ref{def:gen-write}.  Moreover, it is
laborious but straightforward to prove that both $H'_SX$ and $H_SX$ are
conservatively complete $\fot_S$-modules on $X$~\cite{diezel20}, i.e.\
conservatively complete $\fot_S$-modules, equipped with distinguished
maps $\eta\c X\to H'_SX$, $\eta\c X\to H_SX$. In each case $\eta$
sends $x\in X$ to $\brks{\eps,x}$. The partial order on $H'_SX$ (which
we will use for obtaining the least upper bound of a certain sequence
of approximations) is given by the clauses below and relies on the
previous order $\leq$ on trajectories:
\begin{align*}
\infrule{}{\brks{\brks{I,e},x}\appr\brks{\brks{I,e},x}}&&
\infrule{\brks{I,e}\leq\brks{I',e'}}{\brks{I,e}\appr\brks{\brks{I',e'},x}}&&
\infrule{\brks{I,e}\leq\brks{I',e'}}{\brks{I,e}\appr\brks{I',e'}}&&
\end{align*}
The monad given by~\eqref{eq:H-def-pre} admits a sharp
characterization, which is an instance of a general
result~\cite{diezel20}. In more detail,
\begin{proposition}\label{prop:free}
The pair $(H'_SX,\eta)$ is a \emph{free conservatively complete $\fot_S$-module} 
on~$X$, i.e.\ for every conservatively 
complete $\fot_S$-module~$\mm$ and a map $f\c X\to\mm$, there is unique homomorphism
$\hat f\c {H'_SX\to\mm}$ such that~${\hat f\comp\eta = f}$.
\end{proposition}
Intuitively, \Cref{prop:free} ensures that $H'_SX$ is a \emph{least}
conservatively complete $\fot_S$-module generated by $X$. This
characterization entails a construction of an iteration operator
on~$\BBH'_S$ as a least fixpoint. This, in fact, also transfers
to~$\BBH_S$ (as detailed in the proof of the following theorem).
\begin{theorem}\label{thm:HS-Elgot}
Both $\BBH'_S$ and $\BBH_S$ are Elgot monads, for which $f^\istar$ is computed
as a least fixpoint of $\omega$-continuous endomaps $g\mto [\eta,g]^\klstar\comp f$
over the function spaces $X\to H'_S Y$ and $X\to H_S Y$ correspondingly.
\end{theorem}
In this section's remainder, we formally connect the monad $\BBH_S$
with the monad~$\BBH$, the latter introduced in our previous work and
used for providing a semantics to the functional language \hc{}. In
the following section we provide a semantics for the current
imperative language via the monad $\BBH_S$. Specifically, in this
section we will show how to build $\BBH$ from~$\BBH_S$ by considering
additional semantic ingredients on top of the latter.

Let us subsequently write $\eta^S$, $(\argument)_S^\klstar$ and
$(\argument)_S^\istar$ for the unit, the Kleisli lifting and the Elgot
iteration of $\BBH_S$.  Note that $S,X\mto\BBH_S X$ is a
\emph{parametrized monad} in the sense of Uustalu~\cite{uustalu03}, in
particular $H_S$ is functorial in $S$  and for every
${f\c S\to S'}$,~${H_f\c H_S\to H_{S'}}$ is a monad morphism.

Then we introduce the following technical natural transformations 
$\iota\c H_SX\to X\uplus (S\uplus \{\bot\})$ and $\tau\c H_{S\uplus Y} X\to H_SX$.
First, let us define $\iota$:
\begin{align*}
\iota(I,e,x) =&\; 
  \left\{\begin{array}{ll}
        \inr\inl e^0, & \text{if~~} I\neq\emptyset\\
        \inl x, & \text{otherwise}
        \end{array}\right. &
\iota(I,e) = &\;
  \left\{\begin{array}{ll}
        \inr\inl e^0, & \text{if~~} I\neq\emptyset\\
        \inr\inr\bot, & \text{otherwise}
        \end{array}\right.
\end{align*}
In words: $\iota$ returns the initial point for non-zero length
trajectories, and otherwise returns either an accompanying value from
$X$ or $\bot$ depending on that if the given trajectory is convergent
or divergent. The functor $(\argument)\uplus E$ for every $E$ extends
to a monad, called the \emph{exception monad}. The following is easy
to show for $\iota$.
\begin{lemma}\label{lem:iota}
For every $S$, $\iota\c H_S\to (\argument)\uplus (S\uplus \{\bot\})$ is a monad
morphism.
%
%
%
%
\end{lemma}
Next we define $\tau\c H_{S\uplus Y} X\to H_SX$:
\begin{align*}
\tau(I,e,x) = &\;
  \left\{\begin{array}{ll}
        \brks{I,e,x}, & \text{if~~} I = I'\\
        \brks{I',e'}, & \text{otherwise}
        \end{array}\right. &
\tau(I,e) = &\;\brks{I',e'}
\end{align*}
where $\brks{I',e'}$ is the largest such trajectory that for all $t\in I'$, $e^t = \inl e'^{t}$.
\begin{lemma}\label{lem:tau}
For all $S$ and $Y$, $\tau\c H_{ S\uplus  Y} \to H_S$ is a monad morphism.
\end{lemma}
We now arrive at the main result of this section.
\begin{theorem}\label{thm:H-HS}
The correspondence $S\mto H_S S$ extends to an Elgot monad as follows: 
\begin{align*}
\eta(x\in S) =&\; \eta^S(x),\\
(f\c X\to H_S S)^\klstar =&\; \bigl(H_XX\xto{H_{\iota'\comp f} \id} H_{S\uplus\{\bot\}} X\xto{\tau} H_SX\xto{f^\klstar_S} H_SS\bigr),\\
(f\c X\to H_{S\uplus X} (S\uplus X))^\istar =&\; \bigl(X\xto{\!f_{S\uplus X}^\istar\!} H_{S\uplus X} S\xto{\!H_{[\inl, (\iota'\comp f)^\natural]}\id} 
H_{S\uplus \{\bot\}} S\xto{\!\tau\!} H_SS\bigr).
%
\end{align*}
where $\iota' = [\inl,\id]\comp\iota\c H_SS\to S\uplus\{\bot\}$ and 
$(\argument)^\natural\c (X\to (S\uplus X)\uplus\{\bot\})\to (X\to S\uplus\{\bot\})$
is the iteration operator of the maybe-monad $(\argument)\uplus\{\bot\}$ (as in Example~\ref{exp:maybe}).
Moreover, thus defined monad is isomorphic to $\BBH$. 
\end{theorem}
\begin{proof}[Proof Sketch]
It is first verified that the monad axioms are satisfied using abstract properties
of $\iota$ and $\tau$, mainly provided by Lemmas~\ref{lem:iota} and~\ref{lem:tau}.
Then the isomorphism $\theta\c H_SS\iso HS$ is defined as expected:
$\theta([0,d),e,x) = \inl\brks{[0,d],\hat e}$ where $e^t = \hat e^t$
for $t\in [0,d)$, $\hat e^d=x$; and $\theta(I,e) = \inr\brks{I,e}$. It
is easy to see that $\theta$ respects the unit.  The fact that
$\theta$ respects Kleisli lifting amounts to a (tedious) verification
by case distinction. Checking the formula for~$(\argument)^\istar$
amounts to transferring the definition of $(\argument)^\istar$, as
defined in previous work~\cite{GoncharovJakobEtAl18a},
along~$\theta$. See the full proof in the appendix.
\qed\end{proof}


\section{Soundness and Adequacy}\label{sec:deno}

\begin{figure*}[t]
\begin{align*}
\sem{\prog{x\ass t}}(\sigma) =&\; \eta(\sigma\triangledown[\prog{t}\sigma/\prog{x}])\\[0.1ex]
\sem{\prog{\bar{x}' = \bar{u} \> \until \> \prog{t}}}(\sigma) 
    =&\; \brks{[0,\prog{t}\sigma), \l \mathit{t} .\,\sigma\triangledown[\phi_\sigma(\mathit{t})/\bar{\prog{x}}], \sigma\triangledown[\phi_\sigma(\prog{t}\sigma)/\bar{\prog{x}}]}\\[0.1ex]
\sem{\prog{p \scomp \, q}}(\sigma) =&\; \sem{\prog{q}}^\star(\sem{\prog{p}}(\sigma))\\[0.1ex]
\sem{\progife{b}{p}{q}}(\sigma) =&\; \ite{\sem{\prog p}(\sigma)}{\prog{b}\sigma}{\sem{\prog q}(\sigma)}\\[0.1ex]
\sem{\progwhile{b}{p}}(\sigma) =&\;
    (\l\sigma.\, \ite{(\hat H\inr)(\sem{\prog p}(\sigma))}{\prog{b}\sigma}{\eta(\inl\sigma))^\dagger(\sigma)}
\end{align*}   \vspace{-0.8cm}
  \caption{Denotational semantics.}
  \label{det_semantics}
\end{figure*}
Let us start this section by providing a denotational semantics to our
language using the results of the previous section. We will then
provide a soundness and adequacy result that formally connects the
thus established denotational semantics with the operational semantics
presented in Section~\ref{sec:syn_sem}.

First, consider the monad in~\eqref{eq:H-def} and fix
$S=\Reals^{\Vars}$. We denote the obtained instance of $H_{S}$ as
$\hat H$. Intuitively, we interpret a program $\prog{p}$ as a map
$\sem{\mathtt{p}} : S \to \hat H S$ which given an environment (a
map from variables to values) returns a trajectory over $S$.  The
definition of $\sem{\prog{p}}$ is inductive over the structure of
$\prog{p}$ and is given in~\Cref{det_semantics}.

\begin{example}
  Given an element $p = \inl \langle I,e,x \rangle \in \hat H X$, let
  us denote $I$ by $p_\dr$ and $e$ by $p_\ev$, and analogously for
  elements $p = \inr \langle I,e \rangle \in \hat H X$. Now, consider
  the program $\prog{x\ass x+1 \> \scomp \> \blue{wait} \> 1}$ and
  denote its interpretation
  $\sem{\prog{x\ass x+1 \> \scomp \> \blue{wait} \> 1}}$ by
  $f\c S \to \hat H S$. According to the denotational semantics,
  wrapping this program into an infinite while-loop yields
  $(\hat H \inr \comp f)^\dagger\c S \to \hat H S$.  Drawing a
  parallel with \cref{rem:inf}, we will show that we can derive the
  value of the trajectory $(\hat H \inr \comp f)^\dagger(\sigma)$ at
  time instant $\nicefrac{1}{2}$ by unfolding \emph{just once} the
  fixpoint equation concerning $(-)^\dagger$.  First, let us observe
  that
  $((\hat H \inr \comp f)(\sigma))^{\nicefrac{1}{2}}_{\ev} = \sigma
  \triangledown [(\prog{x + 1})\sigma / \prog{x}]$. Moreover, note
  that $((\hat H \inr \comp f)(\sigma))_{\dr} = [0,1]$ because the
  only non-instantaneous term in the program is $\prog{\blue{wait} \> 1}$,
  which terminates after exactly one time unit. Now, according to the
  Kleisli lifting of $\hat H$ (recall Definition~\ref{defn:mon}), and
  since $\nicefrac{1}{2} < 1$, the equation
  \begin{flalign*}
    ((g^\klstar \comp \hat H {\inr} \comp f) (\sigma))^{\nicefrac{1}{2}}_\ev =
    ((\hat H \inr \comp f)(\sigma))^{\nicefrac{1}{2}}_\ev
  \end{flalign*}
  holds for every map $g\c S \uplus S \to \hat HS$. Therefore,
  \begin{flalign*}
    \quad ((\hat H \inr \comp f)^\dagger(\sigma))^{\nicefrac{1}{2}}_\ev & = 
    (([\eta, (\hat H \inr \comp f)^\dagger]^\klstar  \comp
    \hat H \inr \comp f)(\sigma))^{\nicefrac{1}{2}}_\ev & \by{fixpoint equation} \\ & =
    ((\hat H \inr \comp f)(\sigma))^{\nicefrac{1}{2}}_\ev & \by{$\nicefrac{1}{2} < 1$} \\ & =
    \sigma \triangledown [(\prog{x + 1})\sigma / \prog{x}]
  \end{flalign*}
\end{example}
In order to establish soundness and adequacy between the small-step
operational semantics and the denotational semantics, we will use an
auxiliary device. Namely, we will introduce a \emph{big-step}
operational semantics that will serve as midpoint between the two
previously introduced semantics. We will show that the small-step
semantics is equivalent to the big-step one and then establish
soundness and adequacy between the big-step semantics and the
denotational one. The desired result then follows by transitivity.
The big-step rules are presented in Figure~\ref{big_step} and follow
the same reasoning than the small-step ones. The expression
$\prog{p},\sigma,t \Downarrow r, \sigma'$ means that $\prog{p}$
paired with $\sigma$ evaluates to $\prog{r},\sigma'$ at time instant $\prog{t}$.

\begin{figure}
\begin{minipage}{1\textwidth}
\begin{flalign*} 
\lrule{(diff-stop$\Downarrow$)}{\prog{t} < \prog{s}\sigma}{
  \prog{\bar{x}' = \bar{t} \>
  \until \> \prog{s}} \sep\sigma\sep t
    \bsto
    \stp\sep\sigma\triangledown[\phi_\sigma( t)/\bar{\prog{x}}]
}
\end{flalign*} 
\vspace{-4mm}
\begin{flalign*}
\lrule{(diff-skip$\Downarrow$)}{
}{
  \prog{\bar{x}' = \bar{t} \>
  \until \> t} \sep\sigma \sep \prog{t}\sigma
    \bsto
  \skp\sep\sigma\triangledown[\phi_\sigma( \prog{t}\sigma)/\bar{\prog{x}}]
}
\end{flalign*} \vspace{-4mm}
\begin{flalign*}
\lrule{(asg$\Downarrow$)}{}{
  \prog{x\ass t}\sep\sigma\sep 0
    \bsto
  \skp\sep\sigma\triangledown[\prog{t}\sigma/\prog{x}]
}
&&
\lrule{(seq-stop$\Downarrow$)}{\prog{p}\sep\sigma\sep t \bsto \stp\sep\sigma'}{
  \prog{p}\scomp \prog{q}\sep\sigma\sep t  \bsto \stp\sep\sigma'}
\end{flalign*} \vspace{-4mm}
\begin{flalign*}
\lrule{(seq-skip$\Downarrow$)}{\prog{p}\sep\sigma\sep t \bsto \skp\sep\sigma' \qquad 
\prog{q}\sep\sigma'\sep t' \bsto r\sep\sigma''}{
  \prog{p}\scomp \prog{q}\sep\sigma\sep t + t' \bsto r\sep\sigma''}\qquad (\mathtt{r}\in\{\stp,\skp\})
\end{flalign*} \vspace{-4mm}
\begin{flalign*}
  \lrule{(if-true$\Downarrow$)}{\prog{b}\sigma=\top\qquad \prog{p} \sep\sigma\sep t\bsto r\sep\sigma'}{\progife{b}{p}{q}\sep\sigma\sep t \bsto r\sep\sigma'}\qquad (\mathtt{r}\in\{\stp,\skp\})
\end{flalign*} \vspace{-4mm}
\begin{flalign*}
  \lrule{(if-false$\Downarrow$)}{\prog{b}\sigma=\bot\qquad \prog{q} \sep\sigma\sep t\bsto r\sep\sigma'}{\progife{b}{p}{q}\sep\sigma\sep t \bsto r\sep\sigma'}\qquad (\mathtt{r}\in\{\stp,\skp\})
\end{flalign*} \vspace{-4mm}
\begin{flalign*}
\lrule{(wh-true$\Downarrow$)}{b\sigma = \top\qquad \prog{p} \scomp \progwhile{b}{p}\sep\sigma\sep t\bsto r\sep\sigma' }
{
  \progwhile{b}{p}\sep\sigma\sep t 
    \bsto 
  r\sep\sigma'
}\qquad (\mathtt{r}\in\{\stp,\skp\})
\end{flalign*}
\vspace{-4mm}  
\begin{flalign*}
\lrule{(wh-false$\Downarrow$)}{b\sigma = \bot}
{
  \progwhile{b}{p}\sep\sigma\sep 0
    \bsto 
  \skp\sep\sigma 
}
\end{flalign*}
\vspace{-4mm}
  \end{minipage}
  \caption{Big-step Operational Semantics}
  \label{big_step}
\end{figure}
Next, we need the following result to formally connect both styles of
operational semantics.
\begin{lemma}\label{lem:progress}
Given a program $\prog{p}$, an environment $\sigma$ and a time instant $\prog{t}$
\needspace{2\baselineskip}
\begin{enumerate}
  \item if\/ $\prog{p}\sep\sigma\sep t\ssto\prog{p'}\sep\sigma'\sep t'$ and 
$\prog{p}'\sep\sigma'\sep t'\bsto\skp\sep\sigma''$ then 
${\prog{p}\sep\sigma\sep t\bsto\skp\sep\sigma''}$;
  \item if\/ $\prog{p}\sep\sigma\sep t\ssto\prog{p'}\sep\sigma'\sep t'$ and 
$\prog{p}'\sep\sigma'\sep t'\bsto\stp\sep\sigma''$ then 
${\prog{p}\sep\sigma\sep t\bsto\stp\sep\sigma''}$.
\end{enumerate} 
\end{lemma}
\begin{proof}
  The proofs follows by induction over the derivation of the small
  step relation. \qed\end{proof}
\begin{theorem}\label{thm:osem-eq}
The small-step semantics and the big-step semantics are related as follows. Given 
a program $\prog{p}$, an environment $\sigma$ and a time instant $\prog{t}$
\begin{enumerate}
  \item $\prog{p}\sep\sigma\sep t\bsto\skp\sep\sigma'$ iff\/ 
$\prog{p}\sep\sigma\sep t\ssto[\star]\skp\sep\sigma'\sep 0$;
  \item $\prog{p}\sep\sigma\sep t\bsto\stp\sep\sigma'$ iff\/ 
$\prog{p}\sep\sigma\sep t\ssto[\star]\stp\sep\sigma'\sep 0$.
\end{enumerate} 
\end{theorem}
\begin{proof}
  The right-to-left direction is obtained by induction over the length
  of the small-step reduction sequence using \Cref{lem:progress}.
  The left-to-right direction follows by induction over the proof of
  the big-step judgement using Proposition~\ref{prop:shift}.
  \qed\end{proof}
  Finally, we can connect the operational and the denotational
  semantics in the expected way.

\begin{theorem}[Soundness and Adequacy]\label{thm:adeq}
Given a program $\prog{p}$, an environment $\sigma$ and a time instant $\prog{t}$
\begin{enumerate}
  \item $\prog{p}\sep\sigma\sep t\ssto[\star]\skp\sep\sigma'\sep 0$ 
iff\/ $\sem{\prog{p}}(\sigma) = (h\c [0,t)\to \Reals^\Vars,\sigma')$; 
\item $\prog{p}\sep\sigma\sep t\ssto[\star]\stp\sep\sigma'\sep 0$ iff\/
  either
  $\sem{\prog{p}}(\sigma) = (h\c [0,t')\to
  \Reals^\Vars,\sigma'')$ or
  $\sem{\prog{p}}(\sigma) = h\c [0,t')\to \Reals^\Vars$, and in
  either case with $\prog{t}'> \prog{t}$ and $h(\prog{t}) = \sigma'$.
\end{enumerate}
\end{theorem}
Here, ``soundness'' corresponds to the left-to-right directions of the equivalences
and ``adequacy'' to the right-to-left ones.
\begin{proof}
By Theorem~\ref{thm:osem-eq}, we equivalently replace the goal as follows:
\begin{enumerate}
  \item $\prog{p}\sep\sigma\sep t\bsto\skp\sep\sigma'$ 
iff $\sem{\prog{p}}(\sigma) = (h\c [0,t)\to \Reals^\Vars,\sigma')$; 
\item $\prog{p}\sep\sigma\sep t\bsto\stp\sep\sigma'$ iff
    either
  $\sem{\prog{p}}(\sigma) = (h\c [0,t')\to
  \Reals^\Vars,\sigma'')$ or
  $\sem{\prog{p}}(\sigma) = h\c [0,t')\to \Reals^\Vars$, and in
  either case with $\prog{t}'> \prog{t}$ and $h(\prog{t}) = \sigma'$.
\end{enumerate}
Then the ``soundness'' direction is obtained by induction over the
derivation of the rules in Fig.~\ref{big_step}. The ``adequacy''
direction follows by structural induction over $\prog{p}$; for
while-loops, we call on the fixpoint law
$[\eta,f^\dagger]^\star\comp f= f^\dagger$ of Elgot monads.\qed
\end{proof}

\section{Implementation} \label{sec:arch}

This section presents our prototype implementation -- \Lince\ -- which
is available online both to run in our servers and to be compiled and
executed locally (\url{http://arcatools.org/lince}). Its architecture
is depicted in \Cref{fig:Lince}. The dashed rectangles correspond to
its main components. The one on the left (\textbf{Core engine})
provides the parser respective to the while-language and the engine to
evaluate hybrid programs using the small-step operational semantics of
\Cref{sec:syn_sem}. The one on the right (\textbf{Inspector}) depicts
trajectories produced by hybrid programs according to parameters
specified by the user and provides an interface to evaluate hybrid
programs at specific time instants (the initial environment
$\sigma : \Vars \to \Reals$ is assumed to be the function constant on
zero). As already mentioned, plots are generated by automatically
evaluating at different time instants the program given as input.
Incoming arrows in the figure denote an input relation and
outgoing arrows denote an output relation.
The two main components 
are further explained~below.

\begin{figure*}[t]
  \centering
  \input{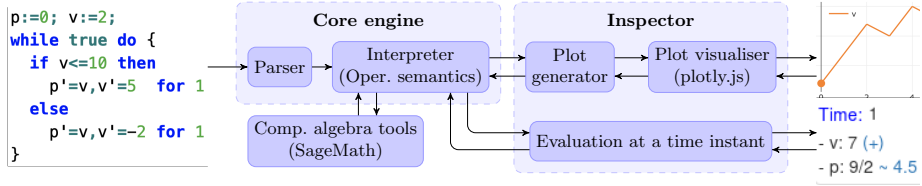}
  \caption{Depiction of \Lince's architecture}
  \label{fig:Lince}
\end{figure*}
\myparagraph{Core engine} 
Our implementation extensively uses the computer algebra tool
\tool{SageMath} \cite{sage}. This serves two purposes: (1) to solve
systems of differential equations (present in hybrid programs); and
(2) to correctly evaluate if-then-else statements. Regarding the
latter, note that we do not merely use predicate functions in
programming languages for evaluating Boolean conditions, essentially
because such functions tend to give wrong results in the presence of
real numbers (due to the finite precision problem). Instead of this,
\Lince\ uses \tool{SageMath} and its ability to perform advanced
symbolic manipulation to check whether a Boolean condition is true or
not. However, note that this will not always give an output,
fundamentally because solutions of linear differential equations
involve transcendental numbers and real-number arithmetic with such
numbers is undecidable \cite{kong15}.  We leave as future work the
development of more sophisticated techniques for avoiding errors in
the computational evaluation of hybrid programs.

\myparagraph{Inspector} The user interacts with \Lince\ at two
different stages: (a) when inputting a hybrid program and (b) when
inspecting trajectories using \Lince's output interfaces.
The latter case consists of adjusting different parameters for
observing the generated plots in an optimal way. Plot parameters
include the time range of observation, visibility of variable's
trajectories, and options to display additional information about the
trajectory (\eg\ where in
time conditional statements are evaluated).

\myparagraph{Event-triggered programs}
Observe that the differential statements
$\prog{x'_1 = t}, \dots, \prog{x'_n = t} \> \prog{\until} \> \prog{t}$
are \emph{time-triggered}: they terminate precisely when the instant
of time $\prog{t}$ is achieved. In the area of hybrid systems it is
also usual to consider \emph{event-triggered} programs: those that
terminate \emph{as soon as} a specified condition $\psi$ becomes true
\cite{witsenhausen66,copp16,Goebel09}. So we next consider atomic
programs of the type
$\> \prog{x'_1 = t}, \dots, \prog{x'_n = t} \> \prog{\uuntil} \>
\prog{\psi}$ where $\psi$ is an element of the free Boolean algebra
generated by $\prog{t} \leq \prog{s}$ and $\prog{t} \geq \prog{s}$
where $\prog{t,s} \in \mathtt{LTerm}(\Vars)$, signalling the
termination of the program. In general, it is impossible to determine
with \emph{exact} precision when such programs terminate (again due to
the undecidability of real-number arithmetic with transcendental
numbers). A natural option is to tackle this problem by checking the
condition $\psi$ periodically, which essentially reduces
event-triggered programs into time-triggered ones. The cost is that
the evaluation of a program might greatly diverge from the nominal
behaviour, as discussed for instance in
documents~\cite{broman18,copp16} where an analogous approach is
discussed for the well-established simulation tools \tool{Simulink}
and \tool{Modelica}. In our case, we allow programs of the form
$\prog{x'_1 = t}, \dots, \prog{x'_n = t} \> \prog{\uuntil}_\epsilon \>
\prog{\psi}$ in the tool and define them as the abbreviation of
$\progwhile{\text{$\neg$} \psi}{\prog{x'_1 = t}, \dots, \prog{x'_n =
    t} \> \prog{\until} \> \epsilon}$. This sort of abbreviation has
the advantage of avoiding spurious evaluations of hybrid programs
w.r.t. the established semantics. We could indeed easily allow such
event-triggered programs natively in our language (\ie\ without
recurring to abbreviations) and extend the semantics accordingly. But
we prefer not to do this at the moment, because we wish first to fully
understand the ways of limiting spurious computational
evaluations arising from event-triggered programs.

\begin{remark}
  \tool{Simulink} and \tool{Modelica} are powerful tools for
  simulating hybrid systems, but lack a well-established, formal
  semantics. This is discussed for example
  in~\cite{bouissou12,foster16}, where the authors aim to provide
  semantics to subsets of \tool{Simulink} and \tool{Modelica}. Getting
  inspiration from control theory, the language of \tool{Simulink} is
  circuit-like, block-based; the language of \tool{Modelica} is
  \emph{acausal} and thus particularly useful for modelling electric
  circuits and the like which are traditionally modelled by systems of
  equations.
\end{remark}

\begin{example}[Bouncing Ball]
  As an illustration of the approach described above for
  event-triggered programs, take a bouncing ball dropped at a
  positive height $\prog{p}$ and with no initial velocity
  $\prog{v}$. Due to the gravitational acceleration $\prog{g}$, it
  falls to the ground and bounces back up, losing part of its
  kinetic energy in the process. This can be approximated by the
  following hybrid program
  \begin{flalign*}
    (\prog{p' = v, v' = g \> \uuntil_{0.01} \> p \leq 0 \wedge v \leq 0)
   \scomp
   (v\ass  v \times - 0.5)}
  \end{flalign*}
  where $\prog{0.5}$ is the dampening factor of the ball. We now want to
  drop the ball from a specific height (\eg\ $\prog{5}$ meters) and let
  it bounce until it stops. Abbreviating the previous program into
  $\prog{b}$, this behaviour can be approximated by
   $\prog{ p\ass  5 \scomp v\ass 0 \scomp }
   \progwhile{true}{b}$.
  \Cref{fig:bb} presents the trajectory generated by the
  ball (calculated by \Lince). Note that since
  $\epsilon = 0.01$ the ball reaches below ground, as shown in
  \Cref{fig:bb} on the right. Other examples of event- and time-triggered
  programs can be seen in \Lince's website.

\end{example}
\begin{figure*}[tb]
  \centering
  \includegraphics[height=29mm]{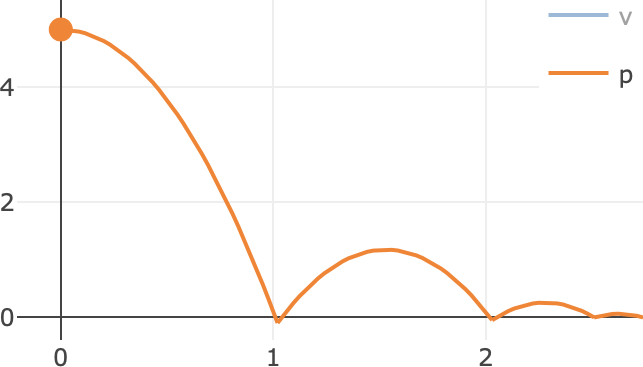}
  ~~~~~~~~~~~~~~
  \includegraphics[height=29mm]{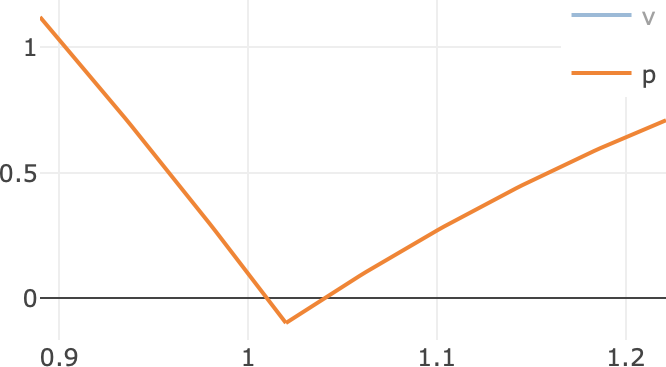}
  \caption{Position of the bouncing ball over time (plot on the left);
    zoomed in position of the bouncing ball at the first bounce (plot
    on the right).}
  \label{fig:bb}
\end{figure*}

\section{Conclusions and future work}
\label{sec:concl}

We introduced small-step and big-step operational semantics for hybrid
programs suitable for implementation purposes and provided a
denotational counterpart via the notion of Elgot monad. These
semantics were then linked by a soundness and adequacy
theorem~\cite{winskel93}. We regard these results as a stepping stone
for developing computational tools and techniques for hybrid
programming; which we attested with the development of \Lince.  With
this work as basis, we plan to explore the following research lines in
the near future.

\noindent
\textbf{Program equivalence}. Our denotational
semantics entails a natural notion of program equivalence
(denotational equality) which inherently includes classical laws of
iteration and a powerful \emph{uniformity}
principle~\cite{SimpsonPlotkin00}, thanks to the use of Elgot
monads. We intend to further explore the equational theory of our
language so that we can safely refactor/simplify hybrid programs. Note that the
theory includes equational schema like
  $(\prog{x\ass a \scomp x\ass b}) ~=~ \prog{x \ass b}$ and
  $(\prog{\blue{wait} \> a \scomp \blue{wait} \> b}) ~=~
  \prog{\blue{wait} \> (a + b)}$
thus encompassing not only usual laws of programming but also
axiomatic principles behind the notion of time.

\noindent
\textbf{New program constructs}. Our while-language is intended to be
as simple as possible whilst harbouring the core, uncontroversial
features of hybrid programming. This was decided so that we could use
the language as both a theoretical and practical basis for advancing
hybrid programming. A particular case that we wish to explore next is
the introduction of new program constructs, including \eg\
non-deterministic or probabilistic choice and exception operations
$\prog{raise(exc)}$. Denotationally, the fact that we used monadic
constructions readily provides a palette of techniques for this
process, \eg\ tensoring and distributive
laws~\cite{Ghani02,manes2007}.

\noindent
\textbf{Robustness}. One important aspect of hybrid programming is
that programs should be \emph{robust}: small variations in their
input should \emph{not} result in big changes in their output
\cite{shorten07,liberzon99}. We wish to extend \Lince\ with
features for automatically detecting non-robust programs. A main
source of non-robustness are conditional statements
$\progife{b}{p}{q}$: very small changes in their input may change the
validity of $\prog{b}$ and consequently cause a switch between
(possibly very different) execution branches. Currently, we are
working on the systematic detection of non-robust conditional
statements in hybrid programs, by taking advantage of the notion of
$\delta$-perturbation \cite{kong15}. 

\subsubsection*{Acknowledgements}
 The first author would like to acknowledge support of German Research Council (DFG) 
 under the project \emph{A High Level Language for Monad-based Processes (GO~2161/1--2)}.
 The second author was financed by the ERDF -- European Regional
 Development Fund through the Operational Programme for Competitiveness
 and Internationalisation -- COMPETE 2020 Programme and by National
 Funds through the Portuguese funding agency, FCT -- Funda\cedilla{c}\~{a}o
 para a Ci\^{e}ncia e a Tecnologia, within project
 POCI-01-0145-FEDER-030947.
 The third author was partially supported by National Funds through \\
 FCT/MCTES, within
 the CISTER Research Unit (UIDB/04234/2020);
 by COMPETE 2020 under the PT2020 Partnership Agreement, through ERDF, and by national funds
 through the FCT, within project POCI-01-0145-FEDER-029946;
 by the Norte Portugal Regional Operational Programme (NORTE 2020)
 under the Portugal 2020 Partnership Agreement, through ERDF and also
 by national funds through the FCT, within project
 NORTE-01-0145-FEDER-028550; and
 by the FCT within project ECSEL/0016/2019 and the ECSEL Joint
 Undertaking (JU) under grant agreement No 876852.  The JU receives
 support from the European Union's Horizon 2020 research and innovation
 programme and Austria, Czech Republic, Germany, Ireland, Italy,
 Portugal, Spain, Sweden, Turkey.

\bibliographystyle{myabbrv}
\bibliography{biblio}

\clearpage
\appendix

\section{Appendix: Omitted Proofs}

\subsection*{Proof of Theorem~\ref{thm:determ}}
The proof follows by inspecting the structure of program terms:
first, for atomic programs the proof follows directly, because the
corresponding premises are mutually exclusive. For conditionals, the
proof also follows directly due to the same reason. For sequential
composition $\prog{ p \> \scomp \> q }$, we need to proceed by case
distinction: if $\prog{p}$ is atomic then the only applicable rules
are \textbf{(seq-stop$^\to$)} and \textbf{(seq-skip$^\to$)} but then
it is easy to see that the corresponding premises are mutually
exclusive.  If $\prog{p}$ is non-atomic then the only applicable
rules are \textbf{(seq$^\to$)} and \textbf{(seq-skip$^\to$)}. But in
this context, the application of \textbf{(seq-skip$^\to$)} requires
that $\prog{p}$ is a while-loop with $\prog{b}\sigma = \bot$ which
forbids the application of \textbf{(seq$^\to$)}. Conversely, the
application of \textbf{(seq$^\to$)} requires that $\prog{p}$ is not
a while-loop with $\prog{b}\sigma = \bot$ and thus we cannot apply
\textbf{(seq-skip$^\to$)}. The proof for while-loops is direct
because the relevant premises are mutually disjoint.
\qed

\subsection*{Proof of Theorem~\ref{thm:HS-Elgot}}
  \Cref{prop:free} entails an enrichment of the Kleisli category of
  $\BBH'_S$ over complete partial orders~\cite[Theorem~7]{diezel20}
  and such monads are Elgot by a general argument~\cite[Theorem
  5.8]{GoncharovSchroderEtAl18}.
The monad $\BBH_S$ can be obtained from $\BBH'_S$ by application of an 
\emph{exception monad transformer}, which sends $\BBT$ to 
$\BBT(\argument\uplus \sum_{I\in [0,\realp]} S^I)$ and then $\BBH_S$ is again Elgot by a 
general result~\cite[Theorem 7.1]{GoncharovSchroderEtAl18}. The obtained iteration
operator is by definition a least fixpoint of the same $\omega$-continuous endomap.
\qed

\subsection*{Proof of Lemma~\ref{lem:tau}}
The proof that the equation concerning monad units
$\tau \comp \eta^{H_{S \uplus Y}} = \eta^{H_{S}}$ holds follows easily
from the fact that
$\tau(\eps,x) = \brks{\eps, x}$.  It remains to
show that the equation concerning Kleisli liftings
\begin{flalign*}
    \tau \comp f^{\klstar}_{H_{S\uplus  Y}} (p) = (\tau \comp f)^\klstar_{H_{S}} \comp \tau (p)
\end{flalign*}
also holds for every element $p \in H_{S\uplus Y} X$.  This is
straightforward, but laborious because it requires several case
distinctions. We first consider the simple case in which $p = \brks{I,e}$
for some interval $I$ and trajectory $e$:
  \begin{align*}
    \tau \comp f^{\klstar}_{H_{S\uplus  Y}} (I,e) =&\; \tau(I,e) \\
    =&\; (\tau \comp f)^\klstar_{H_{S}} \comp \tau (I,e)
  \end{align*}
Next, we consider the case in which $p = \brks{I,e,x}$ for some interval
$I$, trajectory $e$, and element $x \in X$. We proceed with a further
case distinction: first, we assume the existence of some
$t \in I$ such that $e^t = \inr y$ for some $y \in Y$. In this
case $\tau(I,e,x)$ is open convergent, and therefore,
\begin{align*}
  \tau \comp f^{\klstar}_{H_{S\uplus  Y}} (I,e,x) =&\; \tau(I,e,x) \\*
  =&\; (\tau \comp f)^\klstar_{H_{S}} \comp \tau (I,e,x)
\end{align*}
Now we assume the opposite, namely that there is no $t \in I$ such
that $e^t = \inr y$ for some $y \in Y$. For this particular case, we
can slightly abuse notation and state that
$\tau(I,e,x) = \langle I, e, x \rangle$. Going further in case
distinctions, we assume that $f(x) = \langle I', e' \rangle$ for some
interval $I'$ and trajectory $e'$.  In this case,
\begin{align*}
    \tau \comp f^{\klstar}_{H_{S\uplus  Y}} (I,e,x) =&\;
    \tau (\langle I, e \rangle \tconc \langle I', e' \rangle ) \\
    =&\; \langle I, e \rangle \tconc \tau (I',e') \\
    =&\; (\tau \comp f)^\klstar_{H_{S}} \langle I, e, x \rangle \\
    =&\; (\tau \comp f)^\klstar_{H_{S}} \comp \tau ( I, e, x )
\end{align*}
Next we assume that $f(x) = \langle I', e', x' \rangle$ and proceed
with a further case distinction:  we assume the existence of some
$t \in I'$ such that $e'^{t} = \inr y$ for some $y \in Y$, and calculate,
\begin{align*}
    \tau \comp f^{\klstar}_{H_{S\uplus  Y}} (I,e,x) =&\;
    \tau (\langle I, e \rangle \tconc \langle I', e' \rangle, x' ) \\
    =&\; \langle I, e \rangle \tconc \tau (I',e') \\
    =&\; (\tau \comp f)^\klstar_{H_{S}} \langle I, e, x \rangle \\
    =&\; (\tau \comp f)^\klstar_{H_{S}} \comp \tau ( I, e, x )
\end{align*}
Finally, we assume the non-existence of some $t \in I$ such
that $e'^{t} = \inr y$ for some $y \in Y$. For this particular case,
we can slightly abuse notation and state that
$\tau (\langle I, e \rangle \tconc \langle I', e' \rangle, x' ) =
(\langle I, e \rangle \tconc \langle I',e'\rangle, x')$. Then we
obtain,
  \begin{align*}
    \tau \comp f^{\klstar}_{H_{S\uplus  Y}} (I,e,x) =&\;
    \tau (\langle I, e \rangle \tconc \langle I', e' \rangle, x' ) \\
    =&\; (\langle I, e \rangle \tconc \langle I',e'\rangle, x') \\
    =&\; (\tau \comp f)^\klstar_{H_{S}} (I, e, x) \\
    =&\; (\tau \comp f)^\klstar_{H_{S}} \comp \tau ( I, e, x )
  \end{align*}
This concludes the proof.
\qed

\subsection*{Proof of Theorem~\ref{thm:H-HS}}
The fact that $\iota$ and $\tau$ are monad morphisms expands as follows:
\begin{align}
\iota\comp(\eta^S\c X\to H_SX) =&\; \inl,&  \label{eq:iota_eta}\\
\iota\comp (f\c X\to H_SY)^\klstar_S =&\; [\iota\comp f,\inr]\comp\iota.& \label{eq:iota_star}\\
\tau\comp (\eta^{S\uplus Y}\c X\to H_{S\uplus Y} X)  =&\; \eta^S,&  \label{eq:tau_eta}\\
\tau\comp (f\c X\to H_{S\uplus Y}Z)^\klstar_{S\uplus Y} =&\; (\tau\comp f)^\klstar_S\comp\tau.& \label{eq:tau_star}
\end{align}
Note the following simple joint properties of $\iota$ and $\tau$:
\begin{align}
\iota\comp\tau =&\; (\id\uplus [\id\uplus\bot,\inr]) \comp\iota, \label{eq:iota_tau}\\
\tau\comp H_{[f,\inr]}\id=&\;\tau\comp (H_{f}\id)\comp\tau\label{eq:tau_tau}
\end{align}
for any $f\c X\to Y\uplus Z$.

Let us show that
\begin{align}\label{eq:iota-prim}
\iota'\comp f^\klstar = [\iota'\comp f,\inr]\comp\iota'
\end{align}
for $f\c X\to H_S S$. In conjunction with the obvious equation $\iota'\cdot\eta = \inl$
this will certify that~$\iota'$ is a monad morphism.
Indeed,
\begin{flalign*}
\iota'\comp f^\klstar 
=&\; [\inl,\id]\comp\iota\comp f^\klstar_S\comp\tau\comp H_{[\inl,\id]\comp\iota\comp f}\id&\by{definition}\\*
\displaybump=&\; [\inl,\id]\comp[\iota\comp f,\inr]\comp\iota\comp\tau\comp H_{[\inl,\id]\comp\iota\comp f}\id&\by{\eqref{eq:iota_star}}\\
=&\; [\iota'\comp f,\id]\comp\iota\comp\tau\comp H_{[\inl,\id]\comp\iota\comp f}\id\\
=&\; [\iota'\comp f,\id]\comp(\id\uplus [\id\uplus\bot,\inr]) \comp\iota\comp H_{[\inl,\id]\comp\iota\comp f}\id&\by{\eqref{eq:iota_tau}}\\
=&\; [\iota'\comp f,\id]\comp(\id\uplus [\id\uplus\bot,\inr]) \comp (\id\uplus ([\inl,\id]\comp\iota\comp f\uplus\id))\comp\iota\\
=&\; [\iota'\comp f,\id]\comp(\id\uplus [[\inl,\id\uplus\bot]\comp\iota\comp f,\inr])\comp\iota\\
=&\; [\iota'\comp f,\id]\comp(\id\uplus [\iota'\comp f,\inr])\comp\iota&\by{definition}\\
=&\; [\iota'\comp f, [\iota'\comp f,\inr]]\comp\iota\\
=&\; [\iota'\comp f,\inr]\comp [\inl,\id]\comp\iota\\
=&\; [\iota'\comp f,\inr]\comp\iota'&\by{definition}
\end{flalign*}
We proceed with proving the monad laws.
\begin{itemize}
  \item $\eta^\klstar = \id$:  
  $\eta^\klstar = (\eta^S)^\klstar_S\comp\tau\comp H_{[\inl,\id]\comp\iota\comp\eta^S}\id 
  = \tau\comp H_{\inl}\id = \id$, using~\eqref{eq:iota_eta};
  \item $f^\klstar\comp\eta = f$: $f^\klstar\comp\eta = f^\klstar_S\comp\tau\comp H_{[\inl,\id]\comp\iota\comp f}\id\comp\eta^S = 
  f^\klstar_S\comp\tau\comp\eta^{S\uplus\{\bot\}} = f^\klstar_S\comp\eta^{S} = f$, using~\eqref{eq:tau_eta};
  \item $(f^\klstar\comp g)^\klstar = f^\klstar\comp g^\klstar$ where $f\c X\to H_S S$, $g\c Y\to H_XX$:
  the calculation runs as follows.
  \begin{flalign*}
 (f^\klstar\comp g)^\klstar 
=&\; (f^\klstar_S\comp\tau\comp H_{\iota'\comp f}\id\comp g)^\klstar_S
  \comp\tau\comp H_{\iota'\comp f^\klstar\comp g}\id&\by{definition}\\
\displaybump
=&\; f^\klstar_S\comp(\tau\comp H_{\iota'\comp f}\id\comp g)^\klstar_S
  \comp\tau\comp H_{\iota'\comp f^\klstar\comp g}\id&\by{monad law}\\
=&\; f^\klstar_S\comp(\tau\comp H_{\iota'\comp f}\id\comp g)^\klstar_S
  \comp\tau\comp H_{[\iota'\comp f,\inr]\comp\iota' \comp g}\id&\by{\eqref{eq:iota-prim}}\\
=&\; f^\klstar_S\comp(\tau\comp H_{\iota'\comp f}\id\comp g)^\klstar_S
  \comp\tau\comp H_{[\iota'\comp f,\inr]}\id\comp H_{\iota'\comp g}\id&\by{functoriality}\\
=&\; f^\klstar_S\comp (\tau\comp H_{\iota'\comp f}\id\comp g)^\klstar_S\comp\tau\comp H_{\iota'\comp f}\id\comp\tau\comp H_{\iota'\comp g}\id&\by{\eqref{eq:tau_tau}}\\
=&\; f^\klstar_S\comp\tau\comp H_{\iota'\comp f}\id\comp g^\klstar_X\comp\tau\comp H_{\iota'\comp g}\id&\by{\eqref{eq:tau_star}}\\
=&\; f^\klstar\comp g^\klstar.&\by{definition}
  \end{flalign*}
\end{itemize}
The isomorphism $\theta\c H_SS\iso HS$ is defined as expected: 
$\theta([0,d),e,x) = \inl\brks{[0,d],\hat e}$ where $e^t = \hat e^t$ for $t\in [0,d)$ and $\hat e^d=x$;
and $\theta(I,e) = \inr\brks{I,e}$. It is easy to see that~$\theta$ respects unit.
Let us show that $\theta$ also respects Kleisli lifting. Let $f\c X\to H_S S$ and
proceed by case distinction.
\begin{itemize}
  \item We first prove $(\theta\comp f^\klstar)(I,e) = (\theta\comp f)^\klstar(\theta(I,e))$.
  Let $\brks{I',e'}\appr\brks{I,e}$ be the largest such trajectory that 
  for all $t\in I'$, $f(e^t) \neq \eps$. By unfolding definitions, the goal reduces~to 
  \begin{align}\label{eq:alpha-iso1}
  \theta(f^\klstar_S(I',\l t.\, (\snd f(e^t))^0)) = \inr\brks{I'',\l t.\, ((\theta\comp f)(e^t))_\ev^0}.
  \end{align}
  where $I''$ is the largest subinterval of $I$, such that $(\theta\comp f)(e^t)\neq\inr\eps$
  for all $t\in I''$. By definition, $I'=I''$ and $\theta(f^\klstar_S(I',\l t.\, (\snd f(e^t))^0))=
  \brks{I',\l t.\, (\snd f(e^t))^0}$. Therefore, we are left to verify that for all 
  $t\in I'$, $(\snd f(e^t))^0 = (\theta(f(e^t)))_\ev^0$. Indeed, for any $t\in I'$,
  $f(e^t)\neq\eps$ and therefore the initial point of the trajectory returned by 
  $\theta(f(e^t))$ is the same as for $f(e^t)$. 
  \item Next we prove $(\theta\comp f^\klstar)([0,d),e,x) = (\theta\comp f)^\klstar(\theta([0,d),e,x))$,
  which unfolds to 
  \begin{align}\label{eq:alpha-iso2}
  (\theta\comp f^\klstar_S)((\tau\comp H_{\iota'\comp f}) ([0,d),e,x)) = (\theta\comp f)^\klstar(\inl\brks{[0,d],\hat e})
  \end{align}
  where $e^t = \hat e^t$ for $t\in [0,d)$ and $\hat e^d=x$. Again, let $\brks{I',e'}\appr\brks{[0,d),e}$ 
  be the largest such trajectory that for all $t\in I'$, $f(e^t) \neq \eps$. If
  $I'\neq [0,d)$ then~\eqref{eq:alpha-iso2} reduces to~\eqref{eq:alpha-iso1} and we
  are done by the previous clause. Let us proceed under the assumption that $I'=[0,d)$. 
  The left hand side of~\eqref{eq:alpha-iso2} reduces to $\theta(\brks{[0,d),\l t.\, (\snd f(e^t))^0}\later f(x))$.
  If $f(x)=\eps$ then the later reduces $\inr\brks{[0,d),\l t.\, (\snd f(e^t))^0}$
  and since $(\theta\comp f)(x) = \inr\eps$, the right hand side of~\eqref{eq:alpha-iso2}
  also reduces to the same expression by definition of the Kleisli composition of~$\BBH$. 
  Finally, consider the remaining case of $f(x)\neq\eps$. Then 
  $\theta(\brks{[0,d),\l t.\, (\snd f(e^t))^0}\later f(x)) = \inj (\brks{[0,d),\l t.\, (\snd f(e^t))^0}\tconc p)$ where 
  $\theta(f(x)) = \inj p$. Analogously, $(\theta\comp f)^\klstar(\inl\brks{[0,d],\hat e}) = 
  \inj (\brks{[0,d),\l t.\, (\snd f(e^t))^0}\tconc p)$, and we are done.
\end{itemize}
We proceed to verify correctness of the stated characterization of the iteration operator of~$\BBH$.

\begin{lemma}\label{lem:iotap-morph}
The natural transformation $\iota' = [\inl,\id]\comp\iota\c H_SS\to S\uplus\{\bot\}$
is an Elgot monad morphism.
\end{lemma}
\begin{proof}
We have already seen above that $\iota'$ is a monad morphism (the non-trivial 
part of this statement is equation~\eqref{eq:iota-prim}). We are left to check that $\iota'$
is iteration preserving, i.e.\ $(\iota'\comp f^\istar)(x_0) = (\iota'\comp f)^\natural(x_0)$
for all $f\c X\to H_{S\uplus X}(S\uplus X)$ and $x_0\in X$. Let us recall the 
definition of $(\iota'\comp f)^\natural(x_0)$ from Example~\ref{exp:maybe}: we build 
a sequence $x_0,x_1,\ldots$ where $\iota'(f(x_i)) = \inl\inr x_{i+1}$ for every 
$i$, and then $(\iota'\comp f)^\natural(x_0) = \inr\bot$ either if this sequence is infinite 
or $\iota'(f(x_i)) = \inr\bot$ for some $i$, and $(\iota'\comp f)^\natural(x_0) = \inl y$
if $\iota'(f(x_i)) = \inl\inl y$ for some $i$.

Suppose that the constructed sequence is infinite. This means that for every 
$i$, $f(x_i) = \brks{\eps,\inr x_{i+1}}$, or $f(x_i) = \brks{I,e,y}$ with $I\neq\emptyset$,
and $e^0 = \inr x_{i+1}$, or $f(x_i) = \brks{I,e}$ with $I\neq\emptyset$, and $e^0 = \inr x_{i+1}$.
Note also that $\iota'(f(x_i)) = \inr\bot$ for any $i$. Now, consider
\begin{align}\label{eq:iotap-f}
(\iota'\comp f^\istar)(x_0) 
=&\; ([\inl,\id]\comp\iota\comp\tau\comp H_{[\inl, (\iota'\comp f)^\natural]}\id)(f^\istar_{S\uplus X}(x_0)).
\end{align}
If for every $i$, $f(x_i) = \brks{\eps,\inr x_{i+1}}$ then $f^\istar_{S\uplus X}(x_0)=\eps$
and the whole expression~\eqref{eq:iotap-f} evaluates to $\inr\bot$. Otherwise, 
suppose that $f(x_i) = \brks{\eps,\inr x_{i+1}}$ for $i=0,\ldots,n-1$ and either 
$f(x_{n-1}) = \brks{I,e,y}$, $I\neq\emptyset$, $e^0 = \inr x_{n}$, or $f(x_{n-1}) 
= \brks{I,e}$, $I\neq\emptyset$, $e^0 = \inr x_{i+1}$. It is then easy to see 
that in each case either $f^\istar_{S\uplus X}(x_0) = \brks{I,e,y}$ or 
$f^\istar_{S\uplus X}(x_0) = \brks{I,e}$, and in both cases~$I$ non-empty and 
$e^0 = \inr x_n$. The effect of $H_{[\inl, (\iota'\comp f)^\natural]}\id$ on the 
result of $f^\istar_{S\uplus X}(x_0)$ includes replacing $e^0$ with $\inr\bot$,
which in conjunction with subsequent action of~$\tau$, turns the whole expression~\eqref{eq:iotap-f}
into $\inr\bot$.

Suppose next that the sequence $x_0,\ldots$ is finite and prove by induction over
its length~$n$ that $(\iota'\comp f^\istar)(x_0) = (\iota'\comp f)^\natural(x_0)$.
Suppose that $n>0$. Then $\iota'(f(x_0)) = \inl\inr x_{1}$ for a suitable $x_1\in X$, and
\begin{flalign*}
(\iota'\comp f^\istar)(x_0)
=&\; (\iota'\comp [\eta, f^\istar]^\klstar)(f(x_0))&\by{fixpoint law}\\
\displaybump  =&\; [\iota'\comp [\eta, f^\istar],\inr](\iota'(f(x_0)))&\by{\eqref{eq:iota-prim}}\\
=&\; [[\inl, \iota'\comp f^\istar],\inr](\iota'(f(x_0)))\\
=&\; [[\inl, \iota'\comp f^\istar],\inr](\inl\inr x_{1})&\by{assumption}\\
=&\; (\iota'\comp f^\istar)(x_{1})\\
=&\; (\iota'\comp f)^\natural(x_1)&\by{induction hypothesis}\\
=&\; [[\inl,(\iota'\comp f)^\natural], \inr] (\inl\inr x_1)\\
=&\; [[\inl,(\iota'\comp f)^\natural], \inr] (\iota'(f(x_0)))&\by{assumption}\\
=&\; (\iota'\comp f)^\natural(x_0).&\by{fixpoint law}
\end{flalign*}
If $n=0$ then either $\iota'(f(x_0)) = \inl\inl s$ or $\iota'(f(x_0)) = \inr\bot$.
The proof that $(\iota'\comp f^\istar)(x_0) = (\iota'\comp f)^\natural(x_0)$ is then
analogous to the above, except that we need not call the induction hypothesis.
\end{proof}
We proceed with the proof of Theorem~\ref{thm:H-HS}.
In order to show that the iteration operators of $\BBH_S$ and $\BBH$ are connected 
as stated, we first need to recall how the iteration operator of the latter monad
is defined. We do it by resorting to \emph{singular iteration}
and \emph{progressive iteration}~\cite{GoncharovJakobEtAl18a}. Let 
$(\argument)^\iistar$ be the iteration operator transferred from $\BBH$ to $S\mto H_S S$ along the isomorphism $\theta$.

Given $f\c X\to H_{S\uplus X}(S\uplus X)$, we say that $f$ is singular
if for all $x\in X$, $f(x)=\brks{I,e,y}$ implies that $e^t$ factors through $\inl$
for all $t>0$ from $I$, and analogously $f(x)=\brks{I,e}$ implies that $e^t$ factors through $\inl$
for all $t>0$ from $I$.
We say that such $f$ is progressive if 
for every $x\in X$, $f(x)=\brks{I,e,y}$ implies that $e^0$ factors through~$\inl$ unless $I$ is empty;
$f(x)=\brks{\eps,y}$ implies that $y$ factors through~$\inl$;
and $f(x)=\brks{I,e}$ implies that $e^0$ factors through $\inl$ unless~$I$ is empty.
Equivalently, $f$ is progressive if $\iota'\comp f\c X\to (S\uplus X)\uplus\{\bot\}$
factors through $\inl\uplus\id$.  
It is known~\cite[Theorem 20]{GoncharovJakobEtAl18a} that every $f^\iistar$ is 
decomposable as~$g^{\iistar\iistar}$ where the inner iteration is singular, 
the outer iteration is progressive and $g\c X\to H_{(S\uplus X)\uplus X}((S\uplus X)\uplus X)$
is constructed from $f$ in such a way that $f = (H_{[\id,\inr]}[\id,\inr])\comp g$.
%
What we need to show thus is that~$(\argument)^\istar$ agrees with~$(\argument)^\iistar$ both in the 
singular and in the progressive case and that $f^\istar = g^{\istar\istar}$ for 
$f$ and~$g$ as above. For general $g$, the latter equation is called the \emph{codiagonal law} of iteration.
Let us proceed by case distinction.
\begin{itemize}
  \item\emph{(singular case)}~ If $f$ is singular then for a given $x_0\in X$, $f^\iistar(x_0)$ 
  is described as follows. If $(\iota'\comp f)^\natural(x_0) = \inr\bot$ then 
  $f^\iistar(x_0)=\eps$ -- otherwise, we form the longest possible sequence $x_0,\ldots,x_n$ such 
  that $f(x_i) = \brks{\eps,\inl x_{i+1}}$ for all $i$. This sequence, must be finite,
  for otherwise $(\iota'\comp f)^\natural(x_0)$ would be equal to $\inr\bot$, and 
  also $f(x_n)$ cannot be $\eps$ for the same reason. If $f(x_n) = \brks{\eps,\inr y}$ 
  then $f^\iistar(x)=\brks{\eps, y}$. In the remaining cases, $f^\istar_{S\uplus X}(x_n)=\brks{I,e}$ or $f^\istar_{S\uplus X}(x)=\brks{I,e,y}$ 
  with $I\neq\emptyset$, and we put correspondingly $f^\iistar(x_0)=\brks{I,e'}$ and 
  $f^\iistar(x_0)=\brks{I,e',y}$ where $e'$ is calculated as follows:
%
%
%
  \begin{align*}
  e'^t =&\; 
  \left\{\begin{array}{ll}
        s,   & \text{if~~} t=0 \text{~~and~~} (\iota'\comp f)^\natural(x_0) = \inl s\\
        s, & \text{if~~} t>0 \text{~~and~~} e^t    = \inl s
        \end{array}\right. 
  \end{align*}  
Let us go through these clauses and argue that in each clause $f^\istar(x_0)$ is 
defined in the same way. By Lemma~\ref{lem:iotap-morph}, if 
$(\iota'\comp f)^\natural(x_0) = \iota'\comp f^\istar(x_0) = \inr\bot$ then
$f^\istar(x)=\eps=f^\iistar(x)$. We proceed under the assumption that 
$(\iota'\comp f)^\natural(x)\neq\inr\bot$. If $f(x_n) = \brks{\eps,\inr y}$ then
$f^\istar_{S\uplus X}(x_0)$ in the expression 
\begin{align*}
    (\tau\comp H_{[\inl, (\iota'\comp f)^\natural]}\id)(f^\istar_{S\uplus X}(x_0))
\end{align*} 
for $f^\istar(x_0)$, evaluates to $\eta (y)=\brks{\eps,y}$ and hence, the whole 
expression evaluates to $\eta(y)$. Finally, if 
$f^\istar_{S\uplus X}(x_n)=\brks{I,e}$ then $f^\istar_{S\uplus X}(x_0)=\brks{I,e}$
and the effect of $\tau\comp H_{[\inl, (\iota'\comp f)^\natural]}\id$ on~$\brks{I,e}$
in the above expression coincides with the effect described above: for all points 
of $e$, except for the initial one $\tau\comp H_{[\inl, (\iota'\comp f)^\natural]}\id$ simply removes $\inl$,
and for $e^0$ it additionally calls $(\iota'\comp f)^\natural$. But since 
$e^0 = (\iota'\comp f)^\natural(x_n) = (\iota'\comp f)^\natural(x_0)$ and the latter
must be of the form $\inl s$, for the new trajectory $\brks{I,e'}$, $e'^0 = s$. 
The case $f^\istar_{S\uplus X}(x_n)=\brks{I,e,y}$ is analogous.
  \item\emph{(progressive case)}~ The original construction~\cite{GoncharovJakobEtAl18a} 
  of $(\argument)^\iistar$ essentially amounts to regarding the given $f\c X\to H_{S\uplus X}(S\uplus X)$ as
  operating on \emph{partial trajectories} (i.e.\ such pairs $\brks{I,e\c I\to S\uplus X}$
  that $e$ is a partial function), on which the requisite fixpoint is
  calculated as a suitable limit of an $\omega$-chain of approximations, subsequently
  trimmed by discarding trajectory fragments that occur after points of undefinedness. 
  
  For progressive $f$ this construction simplifies: we have 
  $f^\iistar = \hat f_S^\istar$, where $\hat f\c X\to H_S(S\uplus X)$ is manufactured
  from $f$ as follows. 
  Given $x\in X$, if $f(x)=\brks{I,e,y}$ or $f(x)=\brks{I,e}$ and for some $t\in I$,
  $e^t = \inr x'$ and $([\inl,\id]\comp\iota\comp f)(x')=\inr\bot$
  then $\hat f(x) = \brks{I',e'}$ where $I'$ is the largest subinterval of $I$
  that does not contain such $t$ and for every $t\in I'$, $e'^t = s$ where either 
  $e^t = \inl s$ or $e^t = \inr x'$ and $([\inl,\id]\comp\iota\comp f)(x') = \inl\inl s$. 
  Otherwise, i.e.\ if the indicated $t\in I$ does not exist, $\hat f(x) = 
  \brks{I,e',y}$ for $f(x) = \brks{I,e,y}$ and or $\hat f(x)=\brks{I,e'}$ 
  for $f(x) = \brks{e,y}$, where $e'$ is calculated as before under~$I'=I$. 
  
  To show that $f^\iistar=f^\istar$, it therefore suffices to verify that 
  $\hat f$ equals to 
  \begin{align}\label{eq:f-hat}
    \tau\comp H_{[\inl, (\iota'\comp f)^\natural]}\id \comp f.
  \end{align}
  Using progressiveness of $f$, note that for every $x\in X$,
  $(\iota'\comp f)^\natural(x) = \inl y$ if $f(x) = \brks{\eps,\inl y}$, 
  $(\iota'\comp f)^\natural(x) = \inr\bot$ if $f(x) = \eps$, and 
  $(\iota'\comp f)^\natural(x) = \inl e^0$ 
  if $f(x) = \brks{I,e}$ or $f(x) = \brks{I,e,y}$ with $I\neq\emptyset$. By definition, 
  the effect of $H_{[\inl, (\iota'\comp f)^\natural]}\id$ on each $f(x)$ can 
  be described as follows: every $\brks{I,e}\in f(x)$ is sent to $\brks{I,e'}$ and 
  every $\brks{I,e,y}\in f(x)$ is sent to $\brks{I,e',y}$ where $e'$ is obtained from $e$
  by case distinction: 
  \begin{align*}
  e'^t =&\; 
  \left\{\begin{array}{ll}
        e^t, & \text{if~~} e^t \text{~~factors through~} \inl\\
        \inr\bot, &\text{if~~} e^t = \inr x' \text{~~and~~} f(x') = \eps\\
        \inl x'', &\text{if~~} e^t = \inr x' \text{~~and~~} f(x') = \brks{\eps,\inl x''}\\
        u^0, &\text{if~~} e^t = \inr x'\text{~~and~~} f(x') = \brks{J,u} \text{~~and~~} J\neq\emptyset\\
        u^0, &\text{if~~} e^t = \inr x' \text{~~and~~} f(x') = \brks{J,u,x''}\text{~~and~~} J\neq\emptyset
        \end{array}\right. 
  \end{align*}
The effect of $\tau$ in~\eqref{eq:f-hat} on the result amounts to restricting 
the obtained trajectories to subintervals on which the above clause returning
$\inr\bot$ is not effective. The total action of \eqref{eq:f-hat} on $x\in X$
then coincides with the corresponding action of $\hat f$, as described above.  
  \item\emph{(codiagonal law)}~ We will show that $g^{\istar\istar}=(H_{[\id,\inr]}[\id,\inr]\comp g)^\istar$
  for any given $g\c X\to H_{(S\uplus X)\uplus X}((S\uplus X)\uplus X)$. This is 
  one of the axioms of Elgot monads, in particular, it holds for $\BBH_S$, which 
  fact we are going to use in the following calculation. Let 
  \begin{flalign*}
    h =&\; (\iota'\comp H_{[\id,\inr]}[\id,\inr]\comp g)^\natural\c X\to S\uplus\{\bot\},\\
    w =&\; (\iota'\comp  g^\istar)^\natural\c X\to S\uplus \{\bot\},\\
    u =&\; (\iota'\comp  g)^\natural\c X\to (S\uplus X)\uplus \{\bot\}.
  \end{flalign*}
  Then, on the one hand:
  \begin{flalign*}
  (H_{[\id,\inr]}&[\id,\inr]\comp g)^\istar\\ 
  \displaybump  =&\; (\tau\comp H_{[\inl,h]}\id \comp H_{[\id,\inr]}[\id,\inr]\comp g)^\istar_S&\by{definition}\\
    =&\; \tau\comp H_{[\inl,h]\comp[\id,\inr]}\id\comp (H_{\id} [\id,\inr]\comp g)^\istar_{(S\uplus X)\uplus X}\\
    =&\; \tau\comp H_{[[\inl,h],h]}\id\comp (H_{\id} [\id,\inr]\comp g)^\istar_{(S\uplus X)\uplus X}
  \intertext{and on the other hand:}
  g^{\istar\istar}
    =&\; (\tau\comp H_{[\inl,w]}\id\comp (\tau\comp H_{[\inl,u]}\id\comp g)^\istar_{S\uplus X})^\istar_S&\by{definition}\\
    =&\; (\tau\comp H_{[\inl,w]}\id\comp\tau\comp H_{[\inl,u]}\id\comp g^\istar_{(S\uplus X)\uplus X})^\istar_S\\
    =&\; (\tau\comp H_{[[\inl,w],\inr]}\id\comp H_{[\inl,u]}\id\comp g^\istar_{(S\uplus X)\uplus X})^\istar_S&\by{\eqref{eq:tau_tau}}\\
    =&\; (\tau\comp H_{[[\inl,w],\inr]\comp [\inl,u]}\id\comp g^\istar_{(S\uplus X)\uplus X})^\istar_S\\
    =&\; (\tau\comp H_{[[\inl,w],[[\inl,w],\inr]\comp u]}\id\comp g^\istar_{(S\uplus X)\uplus X})^\istar_S\\
    =&\; \tau\comp H_{[[\inl,w],[[\inl,w],\inr]\comp u]}\id\comp (g^\istar_{(S\uplus X)\uplus X})^\istar_{(S\uplus X)\uplus X}\\
    =&\; \tau\comp H_{[[\inl,w],[[\inl,w],\inr]\comp u]}\id\comp (H_{\id} [\id,\inr]\comp g)^\istar_{(S\uplus X)\uplus X}.&\by{codiagonal}
  \end{flalign*}
To obtain the desired equality, we are thus left to show that $h = w$ and $w = {[[\inl,w],\inr]\comp u}$.
Indeed,
\begin{flalign*}
h 
=&\; (\iota'\comp H_{[\id,\inr]}[\id,\inr]\comp g)^\natural&\\
  \displaybump=&\; (([\id,\inr]\uplus\id)\comp\iota'\comp g)^\natural&\by{\eqref{eq:iota_star},\eqref{eq:iota_eta}}\\
=&\; ((\iota'\comp g)^\natural)^\natural&\by{codiagonal law}\\
=&\; (\iota'\comp  g^\istar)^\natural&\by{Lemma~\ref{lem:iotap-morph}}\\
=&\; w
\intertext{and}
w= &\; (\iota'\comp  g^\istar)^\natural\\
= &\; ((\iota'\comp  g)^\natural)^\natural&\by{Lemma~\ref{lem:iotap-morph}}\\
= &\; [[\inl, ((\iota'\comp  g)^\natural)^\natural],\inr]\comp (\iota'\comp  g)^\natural&\by{fixpoint law}\\
= &\; [[\inl, w],\inr]\comp \iota'\comp u.
\end{flalign*}
\end{itemize}

This completes the proof of Theorem~\ref{thm:H-HS}.
\qed  
  
\end{document}